\documentclass[11pt]{article}
\usepackage{mydef2col}
\usepackage{markArticle}

\title{A Generalized Langevin Model of Latent Liquidity and Concave Price Impact}

\shorttitle{GLE Model of Latent Liquidity}

\author{
    \authorstyle{ Andrey Itkin}
    \newline \newline
    \institution{FRE department, Tandon School of Engineering, New York University, email: \url{aitkin@nyu.edu}} \\
}

\date{\today}

\begin{document}

\maketitle

\lettrineabstract{We model market impact as the response to submitted order flow net of counterflow from latent traders, activated when price displacements from the level that would prevail without the order exceed individual thresholds. Order flow depletes this pool, and a generalized Langevin equation governs its recovery over several time scales. Its memory kernels are finite sums of exponentials, so its Markovian lift is exact rather than an approximation. For an undepleted pool, aggregation under explicit assumptions on individual trading responses yields an intermediate square-root regime between linear small- and large-order limits, without imposing a square-root impact law. Scaling thresholds and responses with price noise makes impact in this regime proportional to volatility, and thresholds that grow with the execution horizon make it independent of duration. With constant displayed depth, expected round-trip costs are nonnegative under the log-price convention, independently of the memory. Numerical experiments show that depletion narrows the square-root range and that memory spectra producing similar single-order impacts can respond differently after substantial prior trading. Calibration to market data is left to a companion paper.
}


\section{Introduction} \label{sec:intro}

Market impact is the displacement of an asset price caused by the execution of an order. It is a major cost of trading large positions and a basic measure of how markets absorb demand, and it has therefore drawn attention from practitioners and theorists alike. The empirical regularity that any model must confront is that impact is concave in the size of the order. The displacement produced by a metaorder, an order split into many child trades and worked over minutes to days, grows more slowly than the size, and over a wide range it grows approximately as the square root of the size. This square-root law has been reported across markets, asset classes, and periods and is among the most robust stylized facts of the market-microstructure literature, \cite{AlmgrenEtAl2005, BershovaRakhlin2013, TothEtAl2011}.

The simplest theory of impact is linear. In the model of \cite{Kyle1985}, an informed trader operates against competitive market makers, and the price moves in proportion to signed order flow, with a constant measuring market depth. Linear impact is the natural benchmark and underlies much of the optimal-execution literature, in which a trading schedule is chosen to balance impact cost against timing risk, \cite{AlmgrenChriss2001}. It does not, however, reproduce the observed concavity. Under a time-invariant linear response, the impact of a metaorder is proportional to its size, so the concave behavior seen in the data must arise from a nonlinearity or an adaptation absent from the static linear model.
For the relation between price changes and total traded volume over fixed intervals, \cite{ContKukanovStoikov2014} show that aggregation alone can produce apparent concavity under a linear response to order flow imbalance.

One influential account describes impact as the response of the price to a history of trades. In propagator models, price is a convolution of past order flow with a decaying kernel, so that the contribution of each trade relaxes over time,  \cite{BouchaudEtAl2004}. The decay of the kernel is constrained by no-arbitrage arguments together with the autocorrelation of order flow, \cite{Gatheral2010}. Propagator models capture the transient nature of impact and its relaxation and fit the data well, but the concavity is represented through the response function rather than derived from a mechanism that generates the square-root law. A related strand models order flow itself as a self-exciting (Hawkes) process, reproducing its clustering and long memory, \cite{BacryMastromatteoMuzy2015}. In that setting, no-arbitrage combined with long-memory order flow implies power-law impact and rough volatility \cite{JusselinRosenbaum2020}. In limit-order-book models of optimal execution, the book is consumed by trading and replenished through resilience, \cite{ObizhaevaWang2013,AlfonsiFruthSchied2010}. A viable impact model should not admit price manipulation, that is, a round trip with negative expected cost, \cite{HubermanStanzl2004,AlfonsiSchiedSlynko2012}. Time-varying depth is a known source of such effects, \cite{FruthEtAl2014,AlfonsiAcevedo2014}.

A second line of work locates the concavity in the structure of latent liquidity. Most of the volume available to a large order is not displayed but latent, revealed only as the price moves. If the density of this latent liquidity varies with distance from the current price, a metaorder encounters progressively more opposing liquidity and its impact becomes concave. This class of models provides a mechanism for the square-root impact law and has been studied in latent-order-book and agent-based frameworks, \cite{TothEtAl2011, MastromatteoTothBouchaud2014, DonierEtAl2015}. \cite{BenzaquenBouchaud2017} formulate the dynamics of latent liquidity as a continuous reaction--diffusion process in which liquidity is consumed and replenished, and show how finite cancellation and deposition rates introduce liquidity time scales. Their fractional extension represents a broad distribution of agent time scales and produces both concave impact and slow relaxation of the impact kernel, \cite{BenzaquenBouchaud2018}. These results motivate a more general memory-based description in which liquidity relaxation is not restricted to a single time scale.

The latent-liquidity account can also be read in terms of the flow that a displacement elicits, \cite{Isichenko2021}. A persistent metaorder is detected by other participants, who trade against it, so that the flow the market absorbs is the submitted order net of an induced opposing flow. Latent orders become active only when the price has moved far enough, and the opposing flow is the part of latent liquidity that the displacement has revealed. Concavity then arises when a larger displacement reveals a more than proportionally larger opposing flow, so that the price moves with a shrinking fraction of the order. In this reading the latent order book and the induced counterflow describe one mechanism, and this paper models them as one.

To represent liquidity dynamics over multiple time scales, we use the generalized Langevin equation (GLE), which provides a natural framework for history-dependent dynamics through a memory kernel, \cite{Mori1965,Zwanzig2001,ItkinGLE1}. 
Langevin equations were applied early to price dynamics by \cite{BouchaudCont1998}, with memory kernels in the drift and white noise. In the GLE used here, the memory kernel also fixes the covariance of the colored noise through the fluctuation-dissipation relation. When the kernel is a finite sum of exponentials, the non-local dynamics admit an exact finite-dimensional Markovian lift. The memory is then represented by a set of auxiliary states, making the model straightforward to simulate and calibrate. This construction is well suited to latent liquidity whose response and relaxation extend over multiple time scales.

The contribution of this paper is a dynamical formulation connecting concave market impact to the depletion and recovery of latent liquidity. Threshold-activated counterflow generates an intermediate square-root regime under explicit assumptions, while generalized Langevin dynamics govern the available pool and its dependence on earlier trading. This construction links the impact prefactor, the extent of the concave regime, and post-execution recovery within an exact finite-dimensional Markovian representation, and yields nonnegative round-trip costs under the stated log-price convention. Numerical comparisons further show that memory specifications producing similar isolated-order impacts can generate substantially different responses to a subsequent order after stronger prior trading. These results identify conditional impact as an additional diagnostic of liquidity dynamics beyond the order-size scaling relation.

The numerical study compares four variants: the linear benchmark, a fresh pool that does not deplete, a depleting pool with a single exponential memory mode, and the baseline GLE pool with several modes, to test whether an extended concave impact regime, including an approximately square-root interval, emerges without prescribing a square-root dependence of impact on order size. Impact curves and their local exponents distinguish the effects of counterflow, depletion, and the memory spectrum, while schedule and inherited-state comparisons examine the response to trading history. The calculations combine deterministic benchmarks and Monte Carlo simulation of the lifted dynamics. They illustrate the analytical predictions under the stated specifications.

All reported stochastic results use Monte Carlo simulation of the lifted dynamics. \Cref{app:pinn} records a backward formulation of the same response functionals and a physics-informed neural network (PINN) approximation, verified against Monte Carlo at the baseline, as a starting point for a companion study. That study is reserved for further numerical development, empirical calibration, and validation.

The rest of the paper is organized as follows. \Cref{sec:model} introduces the price equation, the latent dynamics and their Markovian lift, the counterflow of latent counterparties, and the round-trip cost identity. \Cref{sec:benchmark} derives the small-order expansion. \Cref{sec:cf-impact} derives the impact law with a fresh pool, its duration-free variants, and pathwise bounds for a stochastic pool. \Cref{sec:nonlinear,sec:numerical-method} describe the experimental design, numerical methods, and verification scope. \Cref{sec:experiments} reports size--duration, schedule, and inherited-state comparisons and discusses their limitations. \Cref{sec:conclusions} discusses the conclusions, the limitations, and the empirical follow-up. The appendices contain the proofs and the backward PINN formulation.

\section{The model} \label{sec:model}

We now formulate the model. The market is represented by a constant displayed depth and by a latent liquidity state whose dynamics retain the history of past order flow through a GLE. The order moves the price against the depth. The resulting displacement relative to a counterfactual reference reveals latent counterparties who trade against it. The latent state sets the size of the pool from which these counterparties are drawn. We first define the price and the displacement, then introduce the latent dynamics, the counterflow, and the latent pool, and finally show that round-trip costs are nonnegative under the stated log-price convention.

\subsection{The price equation} \label{sec:gle}

Let $S_t$ denote the asset mid-price at time $t\ge 0$ and $X_t = \log S_t$ the log mid-price. Order flow is described by a signed trading rate $q_t$, with $q_t > 0$ for net buying. A metaorder is executed over a horizon $[0,T]$ with total signed volume
\begin{equation} \label{eq:Q}
Q = \int_0^T q_s\,ds.
\end{equation}

Alongside the price we introduce a reference log-price $m_t$, defined as the counterfactual log-price the same asset would follow in the absence of the metaorder, and the displacement the order induces relative to it,
\begin{equation} \label{eq:mispricing}
D_t = X_t - m_t.
\end{equation}
The displacement is measured relative to the counterfactual process $m_t$, not to a fundamental value. In the balanced-market specification considered first, we assume that the counterfactual log-price has no predictable drift under the physical measure $\mathbb P$ and carries the same diffusion as the observed price,
\begin{equation} \label{eq:reference}
dm_t = \sigma _t\,dW_t^X,
\end{equation}
where the ambient volatility $\sigma _t > 0$ and the Brownian motion $W^X$ are those of the price equation \eqref{eq:price} below. The ambient volatility is an adapted stochastic process whose dynamics do not depend on the metaorder. It may follow any stochastic-volatility model, including one with long memory such as the volatility GLE of \cite{ItkinKazbekGLE}, and it may be correlated with the latent-liquidity noise. The reference therefore removes from the price only the drift generated by the order, and the displacement $D_t$ evolves without a diffusion term.

This choice is required for consistency. In the absence of an order the price and its counterfactual must coincide, so ordinary price fluctuations cannot create a displacement to which counterparties respond. A fixed reference, or a reference driven by independent noise, violates this condition and makes the counterflow trade against ordinary volatility. A mean-reverting force toward a fundamental price, as in the fundamentalist term of \cite{BouchaudCont1998}, responds to every price deviation, including those caused by ordinary noise. Such a counterflow would also be economically implausible for two reasons. It would make prices mean-revert predictably in the absence of any metaorder, whereas short-horizon returns are close to uncorrelated. It would also lead counterparties to trade against price changes that carry information. The model assumes instead that counterparties respond to the footprint of the order. They observe the price rather than the footprint itself, and this idealization is discussed in \cref{sec:counterflow}.

The only restriction on volatility is that it is not caused by the order. If the order affects volatility, for instance through the latent state, the construction extends as follows. Let $Y_t^{[0]}$ denote the latent state that would prevail without the order, driven by the same memory noise with $q\equiv 0$, and let the price volatility be a function $\sigma(Y_t)$ of the latent state. The reference is then driven by $\sigma(Y_t^{[0]})$, and the displacement acquires the noise term $\big[ \sigma (Y_t) - \sigma (Y_t^{[0]})\big]\,dW_t^X$. This term is caused by the order and vanishes without it. The extension duplicates the latent state and its memory modes, and it is not pursued in the numerical study.
\Cref{sec:robustness} comments on the treatment of a reference drift.

The state of the market on which trading acts is represented by a scalar latent variable $Y_t$, whose dynamics are governed by the GLE introduced in \cref{sec:lift}. The price responds to order flow through the drift
\begin{equation} \label{eq:price}
dX_t = \frac{q_t - q_t^{\mathrm{cf}}}{L_0}\,dt + \sigma _t\,dW_t^X,
\end{equation}
where $L_0 > 0$ is the displayed depth, $q_t^{\mathrm{cf}}$ is an endogenous counterflow, $\sigma _t > 0$ is the ambient volatility of \eqref{eq:reference}, and $W^X$ is a standard Brownian motion under the measure $\PP$. The depth is the volume that moves the log-price by one unit in the absence of counterflow, as in \cite{Kyle1985}. It is held constant. Empirically, short-interval price changes are linear in the net order flow at the best quotes, with a slope inversely proportional to depth, although depth itself varies intraday \cite{ContKukanovStoikov2014}. \Cref{sec:conclusions} discusses a state-dependent depth.

Price impact arises through a single mechanism, the net flow $q_t - q_t^{\mathrm{cf}}$. A sufficiently persistent order displaces the price, and the displacement reveals latent counterparties who trade against it. The submitted order $q_t$ is partially offset by the counterflow $q_t^{\mathrm{cf}}$, and the price responds to $q_t^{\mathrm{net}} = q_t - q_t^{\mathrm{cf}}$. The pool of latent counterparties is not fixed. Order flow depletes it, and it refills over several time scales. This memory is carried by the latent state $Y_t$, which modulates the counterflow. The qualitative picture is the latent-liquidity one of \cite{BenzaquenBouchaud2018,BenzaquenBouchaud2017}, in which liquidity is consumed by trading and replenished across a broad range of time scales. The realization is different, since the dynamics here are those of the GLE rather than of those models.

The mechanism does not impose a power law on impact. The counterflow of a fresh pool is derived in \cref{sec:counterflow} from the behavior of opposing traders, and the latent pool is specified in \cref{sec:pool}. An approximately square-root law follows from this derivation at intermediate order sizes, as shown in \cref{sec:cf-impact}.

Under a deterministic metaorder and mean-zero noise, the expected impact at time $t$ is
\begin{equation} \label{eq:impact-t}
I(t) = \EE_\PP \left[X_t - X_0\mid q_{[0,t]}\right] = \EE_\PP \left[\int_0^t \frac{q_s - q_s^{\mathrm{cf}}}{L_0}\,ds\right].
\end{equation}

The terminal impact of a metaorder is a functional of the execution schedule,
\begin{equation} \label{eq:impact-QT}
I[q_{[0,T]}] = \EE_\PP\!\left[X_T - X_0 \,\middle|\, q_{[0,T]}\right],
\end{equation}
and depends on the shape of $q_{[0,T]}$, not on the executed size $Q = \int_0^T q_s\,ds$ alone. Two schedules of equal size and duration but different profiles need not give the same impact, since the counterflow depends on the accumulated displacement and on the latent state, both of which carry the history of the order.

The model therefore defines a chain from order flow to impact,
\begin{equation} \label{eq:chain}
q_t\;\longrightarrow\;\big(Y_t,\;D_t\big)\;\longrightarrow\;q_t^{\mathrm{cf}}\; \longrightarrow\; \frac{q_t - q_t^{\mathrm{cf}}}{L_0}\; \longrightarrow\;X_t\;\longrightarrow\;I[q_{[0,T]}],
\end{equation}
in which the latent state and the displacement jointly determine the counterflow that the order meets.

\subsection{Latent-liquidity dynamics and the Markovian lift} \label{sec:lift}

The latent state $Y_t$ is a modeling device rather than a directly observed quantity. Its memory structure is motivated by persistent signed order flow and slowly relaxing impact, \cite{LilloFarmer2004, BouchaudFarmerLillo2009, BouchaudEtAl2004, Gatheral2010}. These observations do not uniquely imply a GLE: a Markovian mean-reverting state already has a finite relaxation time. The present specification allows several relaxation scales and an explicit dependence on past forcing through separate intrinsic and order-flow kernels. A finite spectrum can describe persistence over a prescribed range of horizons without asserting asymptotic long memory. Whether this additional structure improves the description of liquidity dynamics remains a question for numerical comparison and empirical testing.

The dynamics follow an overdamped GLE \cite{Zwanzig2001,Mori1965}. In this construction the present model inherits the finite-exponential Markovian-lift technique developed for the volatility GLE in \cite{ItkinGLE1}, applying it here to latent liquidity rather than volatility.
\begin{equation} \label{eq:gle}
\dot Y_t = - U'(Y_t) - \int_0^t K_{YY}(t - s)\,dY_s + \int_0^t K_{YX}(t - s; Y_s)\, q_s\,ds + \xi _t,
\end{equation}
in which $U$ is a confining potential and $- U'$ is the restoring force. The kernel $K_{YY}$ acts on increments of the latent state and therefore represents retarded friction, while $K_{YX}$ is the cross-memory kernel through which order flow perturbs the state. For paths of finite variation, the intrinsic-memory term may equivalently be written as $\int_0^t K_{YY}(t - s)\dot Y_s\,ds$.

The two kernels need not share a common time scale. The random force $\xi _t$ is a zero-mean stationary Gaussian process whose autocorrelation is fixed by the memory kernel through the second fluctuation-dissipation relation (FDT),
\begin{equation} \label{eq:fdt}
\EE[ \xi _t\, \xi _s] = \sigma _Y^{2}\,K_{YY}(|t - s|),\qquad \sigma _Y\ge 0 .
\end{equation}
The constant $\sigma _Y^{2}$ couples the random force to the retarded friction and sets the fluctuation scale. Relation \eqref{eq:fdt} is imposed on the retarded kernel only. The unit coefficient of $\dot Y_t$ in \eqref{eq:gle} acts as an instantaneous friction without a matching white-noise force. Consequently $\sigma _Y^{2}$ is not the temperature of an equilibrium law, even for a quadratic potential. For a quadratic potential $U(y) = \tfrac12 u_2 y^2$, the linearized stationary variance of $Y_t$ differs from the Gibbs value $\sigma_Y^2/u_2$. Adding the force $\sqrt 2\,\sigma _Y\,dW_t^Y$, driven by an additional Brownian motion $W^Y$, to the latent equation would restore the Gibbs law, at the cost of a nonsmooth latent state. This is not pursued. Relation \eqref{eq:fdt} alone therefore does not determine a particular invariant density. Existence and form of a stationary law remain properties of the complete dynamics. The coupling nevertheless distinguishes \eqref{eq:gle} from a Volterra equation with independently specified noise.

The shared relaxation spectrum permits matching exponential representations of the retarded friction and the colored force. We specify the kernels as finite sums,
\begin{align} \label{eq:Kyy}
K_{YY}(t) &= \sum_{i = 1}^{N} a_i\,e^{ - \gamma _i t}, \qquad a_i, \gamma _i > 0, \\
K_{YX}(t;y) &= \sum_{j = 1}^{M} c_j(y)\,e^{ - \lambda _j t}, \qquad \lambda _j > 0, \label{eq:Kyx}
\end{align}

with $N$ internal modes and $M$ order-flow modes, the latter carrying state-dependent amplitudes $c_j(y)$. In the present model these finite spectra are structural specifications, so the lift is exact. Define stationary OU modes
\begin{equation} \label{eq:etastate}
\eta _{i,t} = \sigma _{Y,i}\int_{ - \infty}^{t}e^{ - \gamma _i(t - s)}\,dW_{i,s}, \qquad \sigma _{Y,i} = \sigma _Y\sqrt{\frac{2 \gamma _i}{a_i}},
\end{equation}
and combine each mode with its friction history,
\begin{equation} \label{eq:hstate}
h_{i,t} = \int_0^t e^{ - \gamma _i(t - s)}\,dY_s - \eta _{i,t}.
\end{equation}
The order-flow memory states are
\begin{equation} \label{eq:gstate}
g_{j,t} = \int_0^t e^{ - \lambda _j(t - s)}\,c_j(Y_s)\,q_s\,ds,
\end{equation}
where the $W_i$ are independent standard Brownian motions. Then $\xi _t = \sum_{i = 1}^N a_i \eta _{i,t}$ satisfies \eqref{eq:fdt} exactly, including at finite times, because the OU modes are initialized in their stationary laws. The lifted states obey
\begin{equation} \label{eq:hg-ode}
dh_{i,t} = - \gamma _i h_{i,t}\,dt + dY_t - \sigma _{Y,i}\,dW_{i,t},\qquad dg_{j,t} = \big[ - \lambda _j g_{j,t} + c_j(Y_t)\,q_t\big]dt,
\end{equation}
and close the latent equation on a finite state without direct white noise on $Y_t$,
\begin{equation} \label{eq:gle-lift}
dY_t = \Big[ - U'(Y_t) - \sum_{i = 1}^{N} a_i h_{i,t} + \sum_{j = 1}^{M} g_{j,t}\Big]dt.
\end{equation}

Thus, $h_{i,t}$ combines the exponentially weighted increments of $Y$ with the corresponding stationary OU force. Although \eqref{eq:hg-ode} contains $dY_t$, the system is explicit after substituting \eqref{eq:gle-lift}. No time history needs to be retained.

In turn, $g_{j,t}$ is purely deterministic conditional on the path of the latent state $Y_t$ and the exogenous order flow $q_t$. It has no independent diffusion term. It acts strictly as an exponentially decaying filter (a leaky integrator) applied to the incoming order flow, with its sensitivity modulated by $c_j(Y_t)$.

This decomposition separates the intrinsic friction and its FDT-consistent fluctuations, represented by the $h$ states, from the response to observed order flow, represented by the conditionally deterministic $g$ states.

Collecting the state as
\begin{equation} \label{eq:Zdef}
\mathbf Z_t = \big(Y_t,h_{1,t},\dots,h_{N,t},g_{1,t},\dots,g_{M,t}\big)^{ \top}\in\RR^{1 + N + M},
\end{equation}
the latent dynamics constitute a finite Markovian system driven by the metaorder, and the lift is exact for the finite kernels.

\begin{proposition}[Exactness of the lift]\label{prop:lift}
For the finite kernels \eqref{eq:Kyy} and \eqref{eq:Kyx}, let $\eta _{i,0}\sim N(0, \sigma _{Y,i}^2/(2 \gamma _i))$ independently, set $h_{i,0} = - \eta _{i,0}$ and $g_{j,0} = 0$, and evolve \eqref{eq:hg-ode}--\eqref{eq:gle-lift}. Then the lifted system reproduces \eqref{eq:gle}, with a stationary Gaussian force satisfying \eqref{eq:fdt} for all $s,t\ge0$.
\end{proposition}

\begin{proof}
See \cref{sec:lift-appendix}.
\end{proof}

The finite spectrum is a structural modeling choice. A positive exponential sum is bounded at zero and decays exponentially at long times, whereas a power-law kernel $K(t) \sim t^{ - \beta }$, with $0 < \beta < 1$, is singular at zero and has an algebraic tail. Such a power law can be approximated by positive exponential sums on a prescribed interval bounded away from zero, but the accuracy depends on the selected rates, weights, and interval, see \cite{ItkinGLELSV}. Here the relaxation modes are specified directly, so the lift is exact for the chosen model and is not a truncation of an assumed power law. Persistence across a finite range of horizons therefore does not assert asymptotic long memory.

This choice differs from the structural short- and long-time kernel scaling in the volatility GLE of \cite{ItkinKazbekGLE}. Both formulations use intrinsic memory acting on increments of the latent state and colored noise coupled to that kernel through \eqref{eq:fdt}. They differ in the specification of the relaxation spectrum.

The potential $U$ may be quadratic, or nonlinear where replenishment is asymmetric or state dependent,
\begin{equation} \label{eq:quartic}
U(y) = \tfrac{u_2}{2}y^2 + \tfrac{u_3}{3}y^3 + \tfrac{u_4}{4}y^4,\qquad u_2 > 0,\quad u_4 > 0,
\end{equation}
the cubic coefficient $u_3$ introducing an asymmetry between buy and sell displacements and the quartic coefficient $u_4$ providing confinement. The latent state is measured from the local minimum of $U$ at zero.

\subsection{Latent counterparties and the counterflow} \label{sec:counterflow}

The counterflow responds to the displacement $D_t = X_t - m_t$ relative to the reference price. It can persist after execution ends and therefore contribute to price reversal, as examined in \cref{sec:relaxation}. The interpretation of a moving reference requires a consistent treatment of background drift in both price processes. \Cref{sec:robustness} discusses this qualification. We first construct the counterflow of a fresh pool, whose intensity does not depend on past order flow, and introduce depletion of the pool in \cref{sec:pool}.

The counterflow of a fresh pool is a signed function of the displacement,
\begin{equation} \label{eq:qcf-fresh}
q^{\mathrm{cf}} = \mathcal A(D),\qquad \operatorname{sgn}\mathcal A(D) = \operatorname{sgn}D,
\end{equation}
with $\mathcal A$ odd and $|\mathcal A|$ nondecreasing in $|D|$. Its sign is that of the displacement rather than that of the order, so that a buy metaorder, which raises the price above its reference, produces a positive counterflow that is subtracted from the order. When displacement has the order's sign, counterflow reduces the net rate $q_t^{\mathrm{net}} = q_t - q_t^{\mathrm{cf}}$. The extent of this reduction depends on the execution duration and response parameters.

We obtain the aggregate function $\mathcal A$ by integrating an assumed individual response rule over the counterparties. Consider a population of potential counterparties. A counterparty acts only when the magnitude of the displacement exceeds an individual threshold $\chi \ge 0$. The threshold reflects the cost of detecting the order and trading against it, including transaction costs, risk limits, and uncertainty about whether the displacement is transient. Counterparties observe the price rather than the displacement itself. They must distinguish the footprint of the order from ordinary price fluctuations, and the resulting uncertainty is a natural source of positive thresholds. Once active, a counterparty trades against the displacement at a rate proportional to the excess $|D_t| - \chi$.
Transaction-cost models motivate the presence of a no-trade region, \cite{Constantinides1986,DavisNorman1990}. They do not by themselves imply a finite trading rate linear in excess displacement: in particular, the Davis--Norman formulation uses singular trading controls at the boundaries of a no-transaction region. The linear rate adopted here is a separate behavioral assumption. Until its threshold is crossed, a counterparty's volume is not displayed. The population is therefore a model of latent liquidity, and the counterflow is the part of it that the displacement has revealed.

Let $\Pi$ be a nonnegative measure on $[0,\infty)$ that describes the distribution of thresholds, weighted by the proportionality constants of the individual trading rates. Its mass is measured in units of volume per unit time per unit of log-price. Aggregating over the population gives
\begin{equation} \label{eq:cf-aggregate}
\mathcal A(D) = \operatorname{sgn}(D)\int_{[0,|D|]}\big(|D| - \chi \big)\, \Pi (d \chi ).
\end{equation}
We write $\Pi (d \chi ) = \pi _0\, \delta _0(d \chi ) + \pi ( \chi )\,d \chi$, where $\delta _0$ is the unit mass at zero. The atom $\pi _0\ge 0$ represents counterparties who act at any nonzero displacement, such as liquidity providers, and the density $\pi \ge 0$ describes counterparties with positive thresholds. Let $| \Pi | = \Pi ([0,\infty))$ denote the total response intensity and $\bar \chi = | \Pi |^{ - 1}\int_{[0,\infty)} \chi \, \Pi (d \chi )$ the mean threshold.

\begin{proposition}[Structure of the counterflow]\label{prop:cf-structure}
Let $\Pi$ have finite mass $| \Pi | > 0$ and finite first moment. Then the following hold.
\begin{enumerate}
\item The function $\mathcal A$ is odd, nondecreasing, and convex on $[0,\infty)$. Its right derivative is $\mathcal A'(D) = \Pi ([0,D])$ for $D\ge 0$.
\item If $\pi _0 = 0$ and $\pi$ is continuous at zero with $\pi (0) > 0$, then $\mathcal A(D) = \tfrac12 \pi (0)\,D|D| + o(D^2)$ as $D\to 0$.
\item As $|D|\to\infty$, $\mathcal A(D) = | \Pi |\operatorname{sgn}(D)\big(|D| - \bar \chi \big) + o(1)$.
\end{enumerate}
\end{proposition}

\begin{proof}
See \cref{sec:cf-proofs}.
\end{proof}

Part (ii) is the central structural property. For the assumed linear individual response, the leading quadratic term depends only on the threshold density at zero. There must be no atom at zero, and the density must be continuous and positive there. Counterparties may have arbitrarily small positive thresholds. What is excluded is a positive response mass at exactly zero. Part (iii) states that finite total response intensity and a finite first threshold moment produce asymptotically linear counterflow. An atom at zero adds the linear term $\pi _0 D$, which dominates the quadratic term for $|D| < 2 \pi _0/ \pi (0)$. The convexity in part (i) distinguishes this construction from a saturating response such as $q_\star\tanh(D_t/d_c)$, which is concave on $[0,\infty)$ and cannot be written in the form \eqref{eq:cf-aggregate}. The consequences for impact are derived in \cref{sec:cf-impact}.

\begin{myremark}[Other onset exponents]\label{rem:onset}
Part (ii) rests on two assumptions. The threshold density is positive at zero, and each counterparty trades at a rate linear in its excess. Both can be relaxed. Suppose that a counterparty with threshold $\chi$ trades at a rate proportional to $\big((|D| - \chi )_ + \big)^n$ with $n > 0$, where $x_ + = \max(x,0)$. Suppose also that $\pi _0 = 0$ and $\pi ( \chi ) = \pi _k \chi ^k\big(1 + o(1)\big)$ as $\chi \to0$, with $k > - 1$ and $\pi _k > 0$. Rescaling the threshold by $|D|$ in \eqref{eq:cf-aggregate} gives
\begin{equation} \label{eq:cf-onset-general}
\mathcal A(D) = \pi _k\,\mathrm{B}(k + 1,n + 1)\,\operatorname{sgn}(D)\,|D|^{p}\big(1 + o(1)\big),\qquad p = k + n + 1,\qquad D\to0,
\end{equation}
where $\mathrm{B}$ is the beta function. Part (ii) is the case $k = 0$, $n = 1$, $p = 2$. An intermediate concave impact law requires $p = k + n + 1 > 1$, which does not follow from $k > - 1$ and $n > 0$ alone. For the idealized global response $\mathcal A(D) = \omega \operatorname{sgn}(D)|D|^p$ with $\omega > 0$ and $p > 1$, the equilibrium under a positive constant rate $q$ is $D_\infty = (q/ \omega )^{1/p}$. With a fresh pool, \eqref{eq:price} and \eqref{eq:reference} give $L_0\dot D_t = q - \mathcal A(D_t)$, and rescaling this equation gives the time scale
\begin{equation} \label{eq:general-onset-time}
t_q = L_0 \omega ^{ - 1/p} q^{1/p - 1}, \qquad q = Q/T.
\end{equation}
Thus, for the global power response at fixed $T$, sufficiently large orders satisfy $T/t_q \gg 1$ and the terminal exponent tends to $1/p$. Suppose instead that the power response is only a local onset approximation, valid for $|D| \ll \chi _c$, where $\chi _c$ is the scale over which the threshold density varies. The conclusion then applies only on an overlap range satisfying
\begin{equation} \label{eq:general-onset-overlap}
T/t_q \gg 1, \qquad (q/ \omega )^{1/p} \ll \chi _c.
\end{equation}
Within such a range the intermediate exponent is $1/(k + n + 1)$. It is one half for a density positive at zero and linear individual demand. Holding $n = 1$ fixed, a density vanishing at zero increases $p$ and lowers the impact exponent. Holding $k = 0$ fixed, convex individual demand has the same effect. The converse changes increase the exponent only within the admissible range $p > 1$. No concave quasi-stationary regime is inferred from this argument when $p \le 1$, or when the two conditions in \eqref{eq:general-onset-overlap} do not overlap.
\end{myremark}

In the numerical study we take an exponential threshold density without an atom,
\begin{equation} \label{eq:cf-exp-density}
\Pi (d \chi ) = \frac{q_\star}{d_c^2}\,e^{ - \chi /d_c}\,d \chi ,\qquad q_\star > 0,\quad d_c > 0,
\end{equation}
which gives
\begin{equation} \label{eq:cf-exp}
\mathcal A(D) = q_\star\operatorname{sgn}(D)\left(\frac{|D|}{d_c} - 1 + e^{ - |D|/d_c}\right).
\end{equation}
Here $d_c$ is the mean threshold and $q_\star/d_c$ is the total response intensity $| \Pi |$. The counterflow is quadratic, $\mathcal A(D)\approx q_\star D|D|/(2d_c^2)$, for $|D|\ll d_c$, and linear, $\mathcal A(D)\approx q_\star\operatorname{sgn}(D)(|D| - d_c)/d_c$, for $|D|\gg d_c$. The dependence on the order-flow history enters through $D_t$, which accumulates the net flow, and through the latent pool of \cref{sec:pool}. Other threshold distributions and individual response rules are admissible.

The thresholds depend on the prevailing volatility. A counterparty must distinguish the footprint of the order from ordinary price fluctuations, and it can do so only when the displacement is large relative to the noise accumulated over its detection horizon $\tau _d$. We therefore measure thresholds and excesses in units of the noise scale $s_t = \sigma _t\sqrt{ \tau _d}$. Let $\hat \Pi$ be a fixed measure on normalized thresholds $\zeta = \chi /s_t$, with mass in units of volume per unit time, and let each counterparty trade at a rate proportional to its normalized excess. The fresh-pool counterflow then takes the form
\begin{equation} \label{eq:cf-scaled}
\mathcal A_t(D) = \hat{\mathcal A}\big(D/s_t\big),\qquad \hat{\mathcal A}(x) = \operatorname{sgn}(x)\int_{[0,|x|]}\big(|x| - \zeta \big)\,\hat \Pi (d \zeta ).
\end{equation}
The normalized function $\hat{\mathcal A}$ has the form \eqref{eq:cf-aggregate}, so \cref{prop:cf-structure} applies to it at every time. For the exponential density, $\hat \Pi (d \zeta ) = (q_\star/c_d^2)\,e^{ - \zeta /c_d}\,d \zeta$ with a dimensionless constant $c_d > 0$, and \eqref{eq:cf-exp} holds with the time-dependent threshold scale
\begin{equation} \label{eq:dc-vol}
d_c(t) = c_d\, \sigma _t\sqrt{ \tau _d}.
\end{equation}
The counterflow thus depends on the displacement measured in units of the prevailing noise. Where no confusion arises, we write $\mathcal A(D)$ for $\mathcal A_t(D)$.

\subsection{The latent pool} \label{sec:pool}

The population of \cref{sec:counterflow} is not inexhaustible. Counterparties that absorbed earlier flow carry inventory and exhaust risk limits, and new counterparties arrive over several time scales. The forcing by submitted order flow in \eqref{eq:gle} represents pressure on available liquidity in reduced form. The latent state is not an inventory variable, and its evolution does not impose a conservation equation between pool depletion and realized counterflow. We represent the current pool size by a multiplicative intensity that depends on the latent state of \cref{sec:lift},
\begin{equation} \label{eq:qcf}
q_t^{\mathrm{cf}} = \rho\big(\operatorname{sgn}(D_t)\,Y_t\big)\,\mathcal A_t(D_t),
\end{equation}
where $\rho : \RR \to (0,\infty)$ is nonincreasing with $\rho(0) = 1$. The order-flow coupling in \eqref{eq:gle} is signed, and the cross-memory amplitudes are positive in the baseline, so buying raises $Y_t$. When $D_t > 0$ the active counterparties are sellers, and $\rho(Y_t) < 1$ records that past buying has depleted them. When $D_t < 0$ the argument is $-Y_t$, and past selling has depleted the buyers. The baseline map is a logistic function with a floor,
\begin{equation} \label{eq:rho}
\rho(y) = \rho_{\mathrm{fl}} + \frac{2\,(1 - \rho_{\mathrm{fl}})}{1 + e^{y/y_\rho}}, \qquad y_\rho > 0, \quad 0 \le \rho_{\mathrm{fl}} < 1,
\end{equation}
which is smooth, takes values in $(\rho_{\mathrm{fl}}, 2 - \rho_{\mathrm{fl}})$, and satisfies $\rho(y) + \rho(-y) = 2$. The scale $y_\rho$ sets how far the latent state must move before depletion becomes appreciable. The floor $\rho_{\mathrm{fl}}$ is the share of the pool that never depletes, such as dealers whose capacity one order cannot exhaust. A value above one on the side opposite to recent flow can be read as counterparties that absorbed that flow and now trade to unwind their inventory.

Without a floor, sufficiently strong depletion can weaken the counterflow and produce a convex transition toward the linear response without counterflow. A positive floor prevents the pool intensity from vanishing, but does not by itself guarantee concavity of impact.

Memory enters the counterflow only through the current level of $Y_t$. The counterflow has no memory kernel of its own. At $D_t = 0$ it vanishes whatever the value of $Y_t$. The latent state therefore acts on the price only after a displacement has formed, which is the defining property of latent liquidity. The displayed depth $L_0$ plays the role of the visible book.

With \eqref{eq:qcf}, the price equation \eqref{eq:price} becomes
\begin{equation} \label{eq:combined-drift}
dX_t = \frac{q_t - \rho\big(\operatorname{sgn}(D_t)\,Y_t\big)\,\mathcal A_t(D_t)}{L_0}\,dt + \sigma _t\,dW_t^X,\qquad D_t = X_t - m_t .
\end{equation}
By \eqref{eq:reference}, the displacement solves
\begin{equation} \label{eq:displacement-ode}
dD_t = \frac{q_t - q_t^{\mathrm{cf}}}{L_0}\,dt,\qquad D_0 = 0,
\end{equation}
so the common price diffusion cancels. Randomness can still enter through the latent state and the volatility-dependent threshold scale. Under the fixed-volatility specification used in the numerical study, only the former remains. The latent-liquidity dynamics of \cref{sec:lift} are unchanged. Because the displacement depends on the price, \eqref{eq:combined-drift} is an equation with feedback, solved together with the latent dynamics.

When $\pi _0 = 0$, the counterflow is of second order in the displacement near the reference price. It therefore leaves the linear response to a small order unchanged, consistent with \cref{sec:benchmark}. At larger displacement it grows at most linearly, with slope bounded by $| \Pi |\sup_y \rho(y)$. After execution ceases, the remaining counterflow contributes a restoring drift whose strength recovers with the pool. Whether the expected displacement decays, and on which time scale, is examined in \cref{sec:relaxation}. Neither pathwise convergence under persistent noise nor an extended concave-impact regime is assumed.

Two limits bracket the model. The fresh pool, $\rho\equiv1$, removes depletion, and the counterflow then depends on the displacement only. Setting $\Pi = 0$ removes the counterflow, and the price follows the linear response of \cite{Kyle1985}.

\begin{myremark}[Effective liquidity]\label{rem:effective-liquidity}
At times when $q_t\ne0$ and $q_t - q_t^{\mathrm{cf}}\ne0$, the drift in \eqref{eq:price} can be written as $q_t/L_t^{\mathrm{eff}}$ with
\begin{equation} \label{eq:Leff}
L_t^{\mathrm{eff}} = L_0\,\frac{q_t}{q_t - q_t^{\mathrm{cf}}}.
\end{equation}
Under the restriction
\begin{equation} \label{eq:admissible-counterflow}
0\le \frac{q_t^{\mathrm{cf}}}{q_t} < 1,
\end{equation}
the counterflow opposes but does not reverse the submitted order, and $L_t^{\mathrm{eff}}\ge L_0$, with equality only when $q_t^{\mathrm{cf}} = 0$. Revealed latent liquidity raises the effective liquidity that the order meets, and depletion of the pool lowers it. The effective liquidity is thus an output driven by the latent state. It is defined only during execution. After execution, $q_t = 0$ while $q_t^{\mathrm{cf}}$ may remain nonzero, and the drift is $-q_t^{\mathrm{cf}}/L_0$. The instantaneous drift constrains only the product of pool intensity and threshold response, so an impact curve alone need not identify the pool dynamics.
\end{myremark}

The lifted system is well posed under the baseline assumptions.

\begin{lemma}[Well-posedness]\label{lem:wellposed}
Let $q$ be a bounded measurable deterministic schedule on $[0,H]$, and let the threshold scale $s_t$ be deterministic, continuous, and positive on $[0,H]$. Assume the potential \eqref{eq:quartic}, the finite kernels \eqref{eq:Kyy} and \eqref{eq:Kyx} with bounded Lipschitz amplitudes $c_j$, a bounded pool map $\rho$ with bounded derivative, and a normalized threshold measure $\hat\Pi$ of finite mass. Then, for the initial law of \cref{prop:lift} and $D_0 = 0$, the system \eqref{eq:hg-ode}, \eqref{eq:gle-lift}, \eqref{eq:displacement-ode} with counterflow \eqref{eq:qcf} has a unique global strong solution on $[0,H]$. Moreover, $|D_t|\le L_0^{-1}\int_0^t|q_s|\,ds$ and $\sup_{t\le H}\EE\big[Y_t^4\big] < \infty$.
\end{lemma}

\begin{proof}
See \cref{sec:roundtrip-proof}.
\end{proof}

\subsection{Round-trip cost and absence of price manipulation} \label{sec:no-manipulation}

The counterflow always opposes the displacement that elicits it. This sign structure has a consequence for the cost of trading. We measure execution cost in log-price units, corresponding to a first-order displacement convention. The expected cost of a schedule $q$ on $[0,H]$ relative to the arrival price is
\begin{equation} \label{eq:cost}
\mathcal C[q] = \EE\left[\int_0^H q_t\,(X_t - X_0)\,dt\right].
\end{equation}
A round trip is a schedule with zero net position at the horizon $H$, that is, $\int_0^H q_t\,dt = 0$. Following the round-trip criterion of \cite{HubermanStanzl2004}, we call a schedule manipulative under this cost convention if $\mathcal C[q] < 0$. The result below concerns this specified log-price functional.

\begin{proposition}[Round-trip cost]\label{prop:roundtrip}
Let $q$ be a bounded predictable trading rate on $[0,H]$ with $\int_0^H q_t\,dt = 0$. Assume that $\int_0^{\cdot} \sigma _t\,dW_t^X$ is a square-integrable martingale and that the counterflow satisfies $q_t^{\mathrm{cf}}D_t\ge 0$ for all $t$. Then
\begin{equation} \label{eq:roundtrip-identity}
\mathcal C[q] = \frac{L_0}{2}\,\EE\big[D_H^2\big] + \EE\left[\int_0^H q_t^{\mathrm{cf}}\,D_t\,dt\right] \ge 0 .
\end{equation}
In particular, whenever the dynamics and expectations are well defined, the model \eqref{eq:combined-drift} has nonnegative round-trip costs under \eqref{eq:cost} for any potential, memory kernels, pool map $\rho > 0$, threshold measure, and initial law.
\end{proposition}

\begin{proof}
See \cref{sec:roundtrip-proof}.
\end{proof}

\Cref{lem:wellposed} supplies the well-posedness and finite expectations that \cref{prop:roundtrip} requires for bounded deterministic schedules. The proof of \cref{prop:roundtrip} uses the martingale property of the log reference $m_t$ in \eqref{eq:reference}. Costs measured in price units are not covered and would require a separate argument.

The identity separates two costs. The first is the displacement left at the end of the round trip. The second is the counterflow the trader has faced, which is nonnegative because counterparties always trade against the displacement. Neither term can be made negative by the evolution of the latent state, so the memory spectrum cannot create a profitable round trip. A related result holds in limit-order-book models with nonlinear price impact and exponential resilience \cite{AlfonsiSchied2010}. Both contrast with \cite{Gatheral2010}, where nonlinear impact with exponential decay admits price manipulation. Here the price responds linearly to net flow, and the nonlinearity sits in the restoring counterflow. Transaction-triggered manipulation in the sense of \cite{AlfonsiSchiedSlynko2012} is not addressed.

A state-dependent depth would add a term to \eqref{eq:roundtrip-identity} that can be negative, because a depth that rises while the displacement is large lowers the cost of unwinding. This is the reason the depth is held constant. \Cref{sec:conclusions} returns to this point.

\subsection{A nested hierarchy} \label{sec:hierarchy}

The model admits a nested hierarchy designed to attribute any observed effect to a specific component. The Kyle benchmark sets $\Pi = 0$, so impact is linear and permanent. The fresh pool sets $\rho\equiv1$, so the counterflow depends on the displacement only, and \cref{sec:cf-impact} solves it in closed form for quadratic onset. The single-mode pool keeps \eqref{eq:qcf} but uses one intrinsic mode and one order-flow mode, $N = M = 1$. Its rates are the geometric means of the baseline rates, and its weights preserve the integrated strengths $\int_0^\infty K_{YY}(s)\,ds$ and $\int_0^\infty K_{YX}(s)\,ds$. The GLE pool is the full model. The last two variants differ only in the memory spectrum, so their difference measures what the multi-scale memory contributes.

The numerical study compares these variants under common schedules. A precise comparison across order sizes requires a specified family of execution schedules, which is introduced in \cref{sec:benchmark}.

\section{Small-order impact} \label{sec:benchmark}

We first specify the family of orders whose impact is to be compared. Fix an execution horizon $T$ and a nonnegative deterministic schedule shape $\psi _T : [0,T] \to \RR_ +$ normalized by
\begin{equation} \label{eq:schedule-shape}
\int_0^T \psi _T(s)\, ds = 1.
\end{equation}
For each signed order size $Q$ defined in \eqref{eq:Q}, scale this fixed shape according to
\begin{equation} \label{eq:scaled-schedule}
q_s(Q) = Q \psi _T(s), \qquad 0 \le s \le T.
\end{equation}
Thus, $Q$ changes the size of the metaorder but not its duration or execution profile. The constant-rate schedule used in \cref{sec:experiments} is the special case $\psi _T(s) = 1/T$.

For the family \eqref{eq:scaled-schedule}, define the terminal impact by
\begin{equation} \label{eq:impact-family}
\mathcal I_T(Q) := I\!\left[q_{[0,T]}(Q)\right].
\end{equation}
The local impact exponent is then
\begin{equation} \label{eq:delta-eff}
\delta _{\mathrm{eff}}(Q;T) = \frac{d \log |\mathcal I_T(Q)|}{d \log |Q|},
\end{equation}
whenever $Q \ne 0$ and $\mathcal I_T(Q) \ne 0$. This definition compares orders that differ only in size. A power-law regime corresponds to an approximately constant $\delta _{\mathrm{eff}}(Q;T)$ over a nontrivial range of $Q$, and a square-root regime to the case in which that constant is close to one half.

We now derive the small-order behavior of the model. Let $Y_t^{[0]}$ denote the latent state without the order, driven by the same memory noise with $q\equiv 0$, as in \cref{sec:gle}. Assume $\pi _0 = 0$ and that the density $\hat\pi$ of $\hat\Pi$ is continuous and positive at zero. Part (ii) of \cref{prop:cf-structure} then gives $\mathcal A_t(D) = \omega_t D|D| + o(D^2)$ with $\omega_t = \hat\pi(0)/(2s_t^2)$. Denote the executed fraction of the schedule by
\begin{equation} \label{eq:executed-fraction}
\Psi_T(s) = \int_0^s \psi_T(s')\,ds', \qquad 0\le s\le T,
\end{equation}
and the mean pool intensities on the two sides by $\bar\rho_s^{\pm} = \EE\big[\rho(\pm Y_s^{[0]})\big]$. Assume that the latent state depends smoothly on the order amplitude, so that $Y_t = Y_t^{[0]} + O(Q)$ in mean square, and that the normalized second-order remainders are uniformly integrable after time integration. The required joint moments of the pool intensity and inverse threshold scale are assumed finite. Then, as $Q \to 0$,
\begin{equation} \label{eq:impact-expansion}
\mathcal I_T(Q) = \frac{Q}{L_0} - \frac{Q|Q|}{L_0^3}\int_0^T \EE\!\left[\omega_s\,\rho\big(\operatorname{sgn}(Q)Y_s^{[0]}\big)\right]\Psi_T(s)^2\,ds + o(Q^2).
\end{equation}
The expansion follows from \eqref{eq:displacement-ode}, because $D_t = Q\Psi_T(t)/L_0 + O(Q^2)$, the displacement has the sign of $Q$, and the counterflow is of second order in the displacement. Therefore
\begin{equation} \label{eq:benchmark-exponent}
\lim_{Q \to 0} \delta _{\mathrm{eff}}(Q;T) = 1.
\end{equation}
The first nonlinear correction makes the impact magnitude concave in order-size magnitude for both signs. If $\omega_s$ is deterministic, the joint expectation in \eqref{eq:impact-expansion} reduces to $\omega_s\bar\rho_s^{\operatorname{sgn}Q}$. With stochastic volatility, the correlation between the threshold scale and the pool must be retained. Under the baseline, the potential is even, the initial law is symmetric, and $\rho(y) + \rho(-y) = 2$. The no-order law of $Y_t$ is then symmetric, so $\bar\rho_s^{\pm} = 1$ and the correction is identical for the two signs. For the constant-rate schedule with fixed thresholds, $\omega_s = \omega$ and $\int_0^T\Psi_T(s)^2\,ds = T/3$. The correction then reduces to $-\omega Q|Q|T/(3L_0^3)$. For the exponential threshold density and smooth pool coupling used at the baseline, depletion by the order enters at cubic order, through the slope of $\rho$ and the response of the latent state, together with the cubic term of the counterflow onset.

\section{Impact with a fresh pool} \label{sec:cf-impact}

The construction of \cref{sec:counterflow} yields an explicit impact law when the pool is fresh, $\rho\equiv1$. Hold the ambient volatility fixed over the execution, $\sigma _t\equiv \sigma$, so that the threshold scale $s = \sigma \sqrt{ \tau _d}$ is constant, and consider the constant-rate schedule $q_t = Q/T$ on $[0,T]$ with $Q > 0$. By \eqref{eq:displacement-ode}, the displacement is then deterministic. It starts at zero and solves
\begin{equation} \label{eq:cf-ode}
L_0\,\frac{dD_t}{dt} = \frac{Q}{T} - \mathcal A(D_t),\qquad D_0 = 0.
\end{equation}
Suppose that $\mathcal A$ is strictly increasing on $[0,\infty)$ and unbounded, as it is for \eqref{eq:cf-exp}. Let $D_\infty(q)$ denote the stationary displacement defined by $\mathcal A(D_\infty) = q$. The displacement increases monotonically toward $D_\infty(Q/T)$ and never exceeds it. Since the counterflow is nonnegative along this path, the terminal impact satisfies
\begin{equation} \label{eq:bounded-counterflow}
0\le\mathcal I_T(Q)\le\min\left(\frac{Q}{L_0},\,D_\infty(Q/T)\right).
\end{equation}
Negative orders follow by symmetry.

The quadratic onset of \cref{prop:cf-structure} determines the impact law when the displacement remains small relative to the thresholds.

\begin{proposition}[Impact under quadratic counterflow]\label{prop:sqrt-cf}
Let $\mathcal A(D) = \omega D|D|$ with $\omega > 0$. Then the solution of \eqref{eq:cf-ode} gives
\begin{equation} \label{eq:sqrt-cf}
\mathcal I_T(Q) = \sqrt{\frac{Q}{ \omega T}}\,\tanh\!\left(\frac{\sqrt{ \omega QT}}{L_0}\right).
\end{equation}
Consequently, $\delta _{\mathrm{eff}}(Q;T)\to 1$ as $Q\to 0$ and $\delta _{\mathrm{eff}}(Q;T)\to\tfrac12$ as $Q\to\infty$.
\end{proposition}

\begin{proof}
The proof is given in \cref{sec:proof4}.
\end{proof}

The two limits have a direct interpretation. A small order ends before the counterflow has grown enough to oppose it, and its impact is the linear value $Q/L_0$. Expanding \eqref{eq:sqrt-cf} to second order in $Q$ recovers the correction $-\omega Q|Q|T/(3L_0^3)$ of \eqref{eq:impact-expansion}. A large order reaches the stationary displacement $\sqrt{Q/( \omega T)}$ within the execution horizon, and its impact grows as the square root of its size.

For the aggregate counterflow \eqref{eq:cf-aggregate} with $\pi _0 = 0$, part (ii) of \cref{prop:cf-structure} gives $\omega = \pi (0)/2$. The quadratic form holds for displacements small relative to the scale $\chi _c$ over which $\pi$ varies, and $\chi _c = d_c$ for \eqref{eq:cf-exp-density}. The square-root law $\mathcal I_T(Q)\approx\sqrt{Q/( \omega T)}$ therefore holds for order sizes in the range
\begin{equation} \label{eq:sqrt-range}
\frac{L_0^2}{ \omega T}\ll Q\ll \omega \, \chi _c^2\,T.
\end{equation}
The lower bound is the size at which the counterflow balances the order within the horizon. The upper bound marks the loss of validity of the quadratic onset approximation. At still larger displacement, part (iii) of \cref{prop:cf-structure} yields asymptotically linear counterflow and hence linear large-order scaling at fixed duration. The model does not assume that active counterparties stop trading when this crossover occurs. Define the counterflow time scale
\begin{equation} \label{eq:tau-cf}
\tau _{\mathrm{cf}} = \frac{L_0}{ \omega \, \chi _c},
\end{equation}
the time over which the counterflow removes a displacement of size $\chi _c$. The ratio of the two bounds in \eqref{eq:sqrt-range} is $(T/ \tau _{\mathrm{cf}})^2$. A parametrically separated square-root regime requires executions long relative to $\tau _{\mathrm{cf}}$. The ratio of its asymptotic bounds grows as the square of the duration. The width measured under a specified numerical tolerance need not equal this ratio. For the exponential density, $\omega = q_\star/(2d_c^2)$ and $\tau _{\mathrm{cf}} = 2L_0d_c/q_\star$.

The scaling of thresholds with volatility determines how impact depends on it. Let $\hat \pi$ denote the density of $\hat \Pi$ in \eqref{eq:cf-scaled}. Then $\omega = \hat \pi (0)/(2s^2)$, and the square-root law becomes
\begin{equation} \label{eq:sqrt-vol}
\mathcal I_T(Q)\approx \sigma \sqrt{ \tau _d}\,\sqrt{\frac{2Q}{\hat \pi (0)\,T}}.
\end{equation}
Impact in this regime is proportional to volatility and to the square root of the order size divided by the response-volume scale $\hat \pi (0)T$. This scale is determined by the counterparty response measure. With a fixed detection horizon it has not been identified with traded market volume over a daily or execution horizon. Consequently, the result reproduces square-root size dependence and volatility proportionality under the stated assumptions, and the empirical normalization of the square-root law requires the further assumptions of \cref{rem:duration-free} below. Volatility proportionality requires both thresholds proportional to $\sigma \sqrt{ \tau _d}$ and individual trading rates proportional to excess displacement measured in those units, with the normalized response measure held fixed. Threshold scaling alone is insufficient. For the exponential density, $\hat \pi (0) = q_\star/c_d^2$. The counterflow time scale is then proportional to $\sigma$, so higher volatility narrows the asymptotic overlap range at fixed duration.

The exponent one half is a consequence of the assumed linear individual response and the positive continuous threshold density at zero. It is not fitted directly to an impact curve, but neither is it independent of these behavioral assumptions. The resulting law is an intermediate asymptotic with linear scaling below and above it. A crossover from linear to square-root impact with increasing order size is documented empirically in \cite{BucciEtAl2019}, where quantitative agreement with a latent-order-book theory requires at least two liquidity time scales. An atom $\pi _0 > 0$ adds a linear term near zero, reducing the range over which quadratic counterflow can dominate and potentially eliminating that range.

Bounded counterflow has a different large-order implication. If $0 \le \mathcal A(D) \le q_\star$ for $D \ge 0$, then
\begin{equation} \label{eq:saturating-impact-bounds}
\frac{Q - q_\star T}{L_0} \le \mathcal I_T(Q) \le \frac{Q}{L_0}, \qquad \frac{\mathcal I_T(Q)}{Q} \longrightarrow \frac{1}{L_0} \quad \text{as } Q \to \infty \text{ at fixed } T.
\end{equation}
Thus bounded counterflow cannot sustain sublinear asymptotic scaling at fixed duration. This bound alone does not exclude a finite intermediate interval of concavity. A specified saturating response, such as $q_\star\tanh(D/d_c)$, must be assessed separately for its finite-range behavior.

Within \eqref{eq:sqrt-range}, impact depends on the order through its average rate $Q/T$. At fixed size, with the remaining parameters fixed and both executions inside this range, it decreases as $T^{ - 1/2}$. This is a specific prediction of the fixed-horizon threshold normalization, rather than a duration-independent size law. \Cref{rem:duration-free} below shows that thresholds scaled with the execution horizon remove this dependence. Testing either version empirically requires a joint comparison of size, duration, volatility, and market-volume normalization.

After execution the counterflow drives the displacement back toward the reference. For $\mathcal A(D) = \omega D|D|$ and $D_T > 0$,
\begin{equation} \label{eq:quadratic-relaxation}
D_{T + \tau} = \frac{D_T}{1 + \omega D_T \tau/L_0},\qquad \tau\ge 0.
\end{equation}
The decay is hyperbolic rather than exponential. A counterflow that vanishes to second order at the reference price relaxes slowly, because the restoring drift weakens as the displacement shrinks.

\begin{myremark}[Duration-free thresholds]\label{rem:duration-free}
The fixed detection horizon $\tau _d$ makes the threshold scale independent of the order. Two alternatives tie it to the execution. In the first, the horizon is proportional to the duration, $\tau _d = c_\tau T$ with a constant $c_\tau > 0$. In the second, thresholds are measured against the noise accumulated since the order began, $s_t = \sigma\sqrt{c_\tau t}$. For quadratic onset with a fresh pool and a constant-rate order, define
\begin{equation} \label{eq:x-duration-free}
x = \frac{1}{\sigma L_0}\sqrt{\frac{\hat\pi(0)\,Q}{2c_\tau}} .
\end{equation}
The terminal impact is then
\begin{equation} \label{eq:duration-free}
\mathcal I_T(Q) = \sigma\sqrt{\frac{2c_\tau Q}{\hat\pi(0)}}\,\tanh x
\qquad\text{or}\qquad
\mathcal I_T(Q) = \sigma\sqrt{\frac{2c_\tau Q}{\hat\pi(0)}}\,\frac{I_1(2x)}{I_0(2x)},
\end{equation}
respectively, where $I_0$ and $I_1$ are modified Bessel functions of the first kind. The derivation is given in \cref{sec:duration-free-proof}. Neither expression contains $T$. The linear regime, the crossover size, and the square-root regime are therefore duration-free. Suppose further that the response intensity at zero threshold is proportional to the market volume rate $v$, so that $\hat\pi(0) = k_v v$ with a constant $k_v > 0$. The square-root regime then reads $\sigma\sqrt{2c_\tau/k_v}\,\sqrt{Q/v}$. This is the familiar dependence on volatility, size, and market volume, with prefactor $\sqrt{2c_\tau/k_v}$. Duration returns only at the upper edge of the regime. For the exponential density the quadratic onset requires $Q \ll q_\star T/2$, a bound in the volume the pool can absorb during the execution. The cost is a narrower regime at a given capacity. In both variants the ratio of the asymptotic bounds equals $q_\star^2T/(4c_\tau\sigma^2L_0^2c_d^2)$. This is $(T/ \tau _{\mathrm{cf}})^2$ with $\tau _{\mathrm{cf}}$ evaluated at the threshold scale reached at $T$, and it grows only linearly in $T$ because $\tau _{\mathrm{cf}}$ now grows as $\sqrt T$. The elapsed-time version also approaches the square root only algebraically, since $I_1(2x)/I_0(2x) = 1 - 1/(4x) + O(x^{-2})$. The first variant requires counterparties to know the duration of the order, and the second only its onset. Both are idealizations of the same kind as the counterfactual reference.
\end{myremark}

The two threshold conventions also shape the path during execution. With the fixed detection horizon and a fresh pool, the displacement approaches the balance $\mathcal A(D) = q$ within a time of order $L_0/(2\sqrt{\omega q})$ under quadratic onset. After this transient the path is flat, and the terminal impact depends on the rate $Q/T$. Under elapsed-time scaling, the solution in \cref{sec:duration-free-proof} gives $D_t\approx\sigma\sqrt{2c_\tau qt/\hat\pi(0)}$ once $t\gg T/x^2$. The displacement then grows as the square root of the executed volume $qt$, whatever the rate, and the terminal law is duration-free. These are the two square-root laws that \cite{DurinRosenbaumSzymanski2023} distinguish, with respect to participation rate and with respect to executed volume. Neither convention reproduces both. The numerical study retains the fixed horizon. Transferring its results to elapsed-time thresholds would require specifying the clock after completion and at the arrival of a later order, and these choices enter the recovery and probe experiments directly.

\subsection{Stochastic pool: pathwise bounds and quasi-static impact} \label{subsec:pool-bounds}

With a stochastic pool the displacement equation is no longer autonomous, but its monotone structure still gives pathwise control. For a buy order the displacement stays nonnegative, since its drift at $D = 0$ equals $q_t/L_0\ge0$. The argument of $\rho$ is then $Y_t$.

\begin{proposition}[Pathwise bounds]\label{prop:pool-bounds}
Let $q_t = q \ge 0$ on $[0,T]$, let the threshold scale $s_t$ be deterministic, and let $\mathcal A_t$ be locally Lipschitz in the displacement. Fix a path on which $\rho_{\min}\le \rho(Y_t)\le\rho_{\max}$ for $t\in[0,T]$, with $0\le \rho_{\min}\le\rho_{\max}$. For $r\ge0$ let $D^{(r)}$ solve $L_0\,\dot D^{(r)}_t = q - r\,\mathcal A_t\big(D^{(r)}_t\big)$ with $D^{(r)}_0 = 0$. Then
\begin{equation} \label{eq:pool-bounds}
D_t^{(\rho_{\max})} \le D_t \le D_t^{(\rho_{\min})}, \qquad 0\le t\le T.
\end{equation}
For $\mathcal A_t(D) = \omega D|D|$ with constant $\omega > 0$ and for $r > 0$, \cref{prop:sqrt-cf} gives
\begin{equation*}
D_T^{(r)} = \sqrt{q/(r\omega)}\,\tanh\big(\sqrt{qr\omega}\,T/L_0\big),
\end{equation*}
and $D_T^{(0)} = qT/L_0$.

\end{proposition}

\begin{proof}
See \cref{sec:pool-bounds-proof}.
\end{proof}

The same argument bounds the relaxation after execution. For quadratic onset and $D_T > 0$,
\begin{equation} \label{eq:pool-relaxation-bounds}
\frac{D_T}{1 + \rho_{\max}\,\omega D_T\tau/L_0} \le D_{T+\tau} \le \frac{D_T}{1 + \rho_{\min}\,\omega D_T\tau/L_0}, \qquad \tau\ge0,
\end{equation}
where $\rho_{\min}$ and $\rho_{\max}$ now bound $\rho(Y_t)$ over $[T, T+\tau]$ along the path. A depleted pool slows early reversal, and its recovery accelerates later reversal. Liquidity recovery and price reversal are therefore linked in this model. With the floor of \eqref{eq:rho}, $\rho(Y_t)\ge \rho_{\mathrm{fl}}$ on every path, so the choice $\rho_{\min} = \rho_{\mathrm{fl}}$ makes the upper bounds in \eqref{eq:pool-bounds} and \eqref{eq:pool-relaxation-bounds} hold on all paths.

The bounds become sharp when the pool varies little along the path. A related approximation applies after the initial transient when displacement adjusts rapidly relative to changes in the pool and the execution rate, giving $\rho(Y_t)\,\mathcal A_t(D_t) \approx q$. The baseline counterflow scale $\tau _{\mathrm{cf}} = 0.02\tau _0$ helps locate this regime, but the local adjustment time also depends on the order rate and pool intensity. For quadratic onset and fixed thresholds this gives $D_t \approx \sqrt{q/(\omega\rho(Y_t))}$, and the terminal impact becomes
\begin{equation} \label{eq:quasi-static}
\mathcal I_T(Q) \approx \sigma\sqrt{\tau _d}\,\sqrt{\frac{2Q}{\hat\pi(0)\,T}}\;\EE\big[\rho(Y_T)^{-1/2}\big].
\end{equation}
The square-root law thus acquires the factor $\EE[\rho(Y_T)^{-1/2}]$. The latent state inherited from earlier flow and its evolution during execution determine this factor. By Jensen's inequality it is at least $\EE[\rho(Y_T)]^{-1/2}$. Depletion lowers the mean intensity, and fluctuations of the pool raise the factor further. With the duration-free thresholds of \cref{rem:duration-free}, the same factor multiplies the duration-free law.

Fast displacement adjustment does not imply that the pool has reached equilibrium or that \eqref{eq:quasi-static} is independent of the memory spectrum. The law of the forced state $Y_T$, and hence $\EE[\rho(Y_T)^{ - 1/2}]$, depends on the kernels, their weights, the initialization, and the execution history. For a frozen pool intensity $r > 0$ and constant rate $q > 0$, linearization about the balance $r\mathcal A(D_*) = q$ gives the local displacement relaxation time
\begin{equation} \label{eq:local-displacement-time}
\tau _D = \frac{L_0}{r\mathcal A'(D_*)} = \frac{L_0}{2\sqrt{r\omega q}} \qquad \text{for } \mathcal A(D) = \omega D^2,\quad D > 0.
\end{equation}
Thus a comparison with the kernel time scales must account for the execution rate and pool state, as well as the coupled latent dynamics. The ratio of \eqref{eq:tau-cf} to the shortest kernel time is not a universal criterion for spectral insensitivity.

Matching integrated order-flow strengths does provide one reason for similar limiting responses. With constant amplitudes, constant forcing, and no noise, a stationary state satisfies $h_i^* = 0$, $g_j^* = c_j^{(0)}q/\lambda _j$, and $U'(Y_*) = q\sum_j c_j^{(0)}/\lambda _j$. The matched spectra therefore have the same forced equilibrium condition, although their transients and stochastic marginal laws can differ. Conversely, for the atom-free exponential response used here, reducing $q_\star$ to zero at fixed thresholds and schedule removes the counterflow and recovers the linear benchmark for every spectrum. Spectral separation need not increase monotonically as the counterflow is slowed. For general schedules, the evolving pool and the execution shape remain coupled.

\section{Scope and design of the simulation study} \label{sec:nonlinear}

The numerical study compares the size--duration response, schedule dependence, and inherited-state effects of the model under prescribed order flow. The calculations comprise analytical and deterministic benchmarks, baseline parameter-targeting and no-order calculations, and Monte Carlo comparisons of the stochastic pool variants. A backward PINN formulation, verified against Monte Carlo at the baseline, is recorded in \cref{app:pinn}.

\subsection{Nonlinear potential} \label{subsec:simulation-potential}

The restoring force of the potential \eqref{eq:quartic} is
\begin{equation} \label{eq:simulation-force}
U'(y) = u_2 y + u_3 y^2 + u_4 y^3.
\end{equation}
The coefficient $u_2$ determines the local curvature, $u_3$ introduces asymmetry, and $u_4$ provides quartic confinement. Linear latent dynamics driven by Gaussian noise have a Gaussian stationary marginal whenever a stationary law exists. Changing the linear memory spectrum alone does not produce a non-Gaussian marginal. The baseline therefore retains symmetric quartic confinement, $u_3 = 0$ and $u_4 > 0$, while the quadratic case is used for analytical checks. No Boltzmann form for the stationary law is assumed.

\subsection{Controlled comparison of the variants} \label{subsec:simulation-hierarchy}

The four variants of \cref{sec:hierarchy} are compared under common schedules and the finite-start initialization convention. The potential, threshold measure, depth, and pool map are unchanged whenever they remain active. Spectral comparisons preserve integrated kernel strengths, as specified in \cref{sec:powerlaws}, and the modal initial laws follow the corresponding spectra. Parameters are not recalibrated after a component is removed. All reported pool comparisons use constant cross-memory amplitudes $c_j(y) = c_j^{(0)}$ and the symmetric baseline potential.

The principal outputs are the expected price response, terminal impact, its local exponent, mean pool intensity, and accumulated counterflow. Maxima of expected trajectories are distinguished from expectations of pathwise maxima. For a prescribed order $q$, define the response relative to an otherwise identical no-order calculation by
\begin{equation} \label{eq:paired-impact}
J_q(t) = \EE\!\left[X_t^{[q]}\right] - \EE\!\left[X_t^{[0]}\right].
\end{equation}
Both expectations use the same initial-law convention and model parameters. For an isolated order started at zero displacement, the no-order displacement and counterflow vanish identically, so $J_q(t) = \EE[D_t^{[q]}]$ and $J_q(T)$ coincides with \eqref{eq:impact-QT}. Control subtraction remains useful for latent-state diagnostics and for incremental responses following a prior order. Common random numbers reduce sampling variability when these differences are estimated by simulation.

\subsection{Parameter scales and experimental specification} \label{subsec:simulation-scales}

Let $\tau _0$ be the reference time, $p_0$ the log-price displacement scale, and $L_0$ the displayed depth, with $Q_0 = L_0p_0$. Sizes, durations, relaxation rates, and counterflow intensity are reported as $Q/Q_0$, $T/\tau _0$, $\gamma_i\tau _0$ and $\lambda_j\tau _0$, and $q_\star\tau _0/Q_0$. The latent-state normalization, baseline parameters, and initial law are specified in \cref{subsec:baseline-parameterization}.

The baseline uses two intrinsic-memory modes and two order-flow modes. Ambient volatility is fixed during each execution, $\sigma_t \equiv \sigma_X$, so it enters the displacement dynamics through the threshold scale \eqref{eq:dc-vol}. Fixing ambient volatility isolates the pool dynamics.

The displayed impact curves and the band-width table use $81$ logarithmically spaced order sizes $V = |Q|$ on $10^{ - 4} \le V/Q_0 \le 10^4$. The curves are shown for $T/\tau _0 \in \{1,3,10\}$, and the table uses $T/\tau _0 \in \{0.1,0.3,1,3,10,30\}$. Schedule comparisons use $Q = Q_0$, $T = \tau _0$, and pause fraction $\kappa = 0.3$. The baseline inherited-state comparison uses prior and probe orders of size $Q_0$ and duration $\tau _0$, with three short gaps. The additional comparison in \cref{sec:inherited-depletion} increases the prior size to $3Q_0$ and $5Q_0$, keeps the probe and both durations fixed, and uses nine gaps from zero to $50\tau _0$. The broad-spectrum specification has $N = M = 4$ and $R_\gamma = R_\lambda = 1000$.

The pool-variant curves and the stochastic band-width estimates reported in \cref{sec:experiments} use $2048$ paths and a time step of $0.01\tau _0$. The Kyle response is analytical, and the fresh-pool response is obtained from \eqref{eq:cf-ode}. The quadratic-onset closed form \eqref{eq:sqrt-cf} is a separate benchmark and is not substituted for the exponential-threshold response in the numerical comparisons.

\section{Numerical method} \label{sec:numerical-method}

The finite-exponential kernels provide an exact Markovian representation of the model. The reported stochastic comparisons use direct Monte Carlo simulation of this lifted system, while the fresh-pool benchmark is computed by ordinary differential equation integration. \Cref{app:pinn} gives backward equations for the same conditional expectations and accumulated functionals, with a PINN approximation verified against Monte Carlo at the baseline.

The Monte Carlo implementation uses a drift-implicit midpoint scheme, with schedule switches included as integration boundaries. Displacement and accumulated counterflow use the same midpoint counterflow, preserving discrete volume balance up to the nonlinear-solver tolerance. Each midpoint solve uses a relative root tolerance of $10^{-11}$ and at most $60$ nonlinear iterations. The baseline simulation settings and the verification checks are reported in \cref{subsec:baseline-parameterization,subsec:test-verification}.

\subsection{Lifted numerical system} \label{subsec:numerical-lift}

For a prescribed deterministic order schedule $q$, let the augmented state be
denoted by $\widetilde{\mathbf Z}_t = (D_t, Y_t, h_{1,t}, \ldots, \allowbreak h_{N,t},g_{1,t}, \ldots, g_{M,t})$. The lifted dynamics can be written in the form
\begin{align} \label{eq:numerical-displacement}
dD_t &= b_q(t,\widetilde{\mathbf Z}_t) \, dt, \\
dY_t &= F(t,\widetilde{\mathbf Z}_t) \, dt, \nonumber \\
dh_{i,t} &= \left[F(t,\widetilde{\mathbf Z}_t) - \gamma _i h_{i,t}\right] \, dt - \sigma _{Y,i} \, dW_{i,t}, \qquad i = 1,\ldots,N, \nonumber \\
dg_{j,t} &= \left[ - \lambda _j g_{j,t} + c_j(Y_t) q_t\right] \, dt, \qquad j = 1,\ldots,M. \nonumber
\end{align}
Here, $F$ denotes the latent-state drift obtained from the Markovian representation of the generalized Langevin equation, while $b_q = \big[q_t - \rho(\operatorname{sgn}(D)\,y)\,\mathcal A_t(D)\big]/L_0$ is the displacement drift containing the prescribed order and the counterflow modulated by the latent pool. Because the price and the reference share the same diffusion \eqref{eq:reference}, the price noise cancels in the displacement and does not enter the lifted system. Its correlation with the latent-liquidity noise affects the joint law of the price and the latent state, but not the order-induced responses under the fixed-volatility specification used below.

When $c_j$ is constant and $g_{j,0}$ is deterministic, the corresponding order-flow memory state can be evaluated analytically,
\begin{equation} \label{eq:numerical-deterministic-flow-memory}
g_j(t) = e^{ - \lambda _j t} g_{j,0} + c_j \int_0^t e^{ - \lambda _j(t - s)} q_s \, ds.
\end{equation}
In this case $g_j$ is a deterministic function of time.

\subsection{Common computational protocol} \label{subsec:computational-protocol}

Each comparison specifies the model vector
\begin{equation} \label{eq:parameter-vector}
\Theta = \left(U, L_0, \rho, \{a_i,\gamma_i\}_{i = 1}^N, \{c_j,\lambda_j\}_{j = 1}^M, \hat\Pi, \tau_d, \sigma_X, \sigma_Y, \nu_0\right),
\end{equation}
together with the order schedule, execution horizon, observation horizon, and numerical settings. Only the components identified as controls are changed. The reported Monte Carlo estimates are formed from the simulated response functionals. Order and control calculations use the same parameter values and initialization convention, with common random numbers where paired differences are estimated by simulation.

Primary outputs are impact, pool intensity, and accumulated counterflow. Local exponents, execution-weighted displacement, relaxation ratios, and crossing times are then computed from these outputs. \Cref{subsec:test-verification} describes the analytical benchmarks and numerical checks.

\subsection{Baseline parameterization} \label{subsec:baseline-parameterization}

The baseline parameters are fixed by stated targets in the reference units of \cref{subsec:simulation-scales} before any impact curve is computed. \Cref{tab:baseline-parameters} collects the values and the basis of each choice. Calibration to market data is deferred to the companion study.

The reference units are $\tau _0 = p_0 = L_0 = 1$, so that $Q_0 = L_0 p_0 = 1$. The depth $L_0$ is the displayed depth of \eqref{eq:price}. The price scale is tied to volatility through $p_0 = \sigma _X\sqrt{ \tau _0}$, which gives $\sigma _X = 1$. Impact is then measured in units of price volatility over the reference time, the normalization used in empirical studies of the square-root law. The latent state is centered at zero, the minimum of $U$, and its unit is fixed by the cross-memory targeting below.

The potential has unit curvature, $u_2 = 1$, so that without memory the latent state relaxes on the time scale $\tau _0$. The quartic coefficient $u_4 = 0.1$ supplies a confining force equal to one tenth of the linear restoring force at $|y| = 1$. The baseline sets $u_3 = 0$ so that buy and sell orders are treated alike.

The memory rates are placed so that the duration grid of \cref{subsec:test-concavity} lies below, within, and above them. The intrinsic rates are $\gamma _i \tau _0 \in \{1, 0.1\}$ and the order-flow rates are $\lambda _j \tau _0 \in \{2, 0.2\}$. Distinct intrinsic and order-flow rates avoid coincidences that are not part of the model. Both spectra have breadth $R_ \gamma = \gamma _{\max}/ \gamma _{\min} = 10$ and $R_ \lambda = \lambda _{\max}/ \lambda _{\min} = 10$. The intrinsic weights give each mode an equal share of the integrated strength $S_{YY}^{\mathrm{int}} = \sum_i a_i/ \gamma _i = 1$. The cross-memory amplitudes have equal ratios $c_j^{(0)}/ \lambda _j$. Their common scale is fixed by one condition: a constant-rate order of size $Q_0$ executed over $\tau _0$ displaces the latent state by one unit at the end of execution in the noise-free lifted dynamics. Solving this condition gives $c^{(0)} = (4.94, 0.494)$, which we round to $(5, 0.5)$. The rounded values give a displacement of $1.01$ units. This choice locates the onset of depletion near $Q_0$. It does not determine the local impact exponent, the width of any approximately square-root regime, or their dependence on duration.

The noise scale $\sigma _Y$ is chosen so that the pool intensity fluctuates by about eight percent in the absence of an order. The stationary covariance of the linearized no-order lifted system gives a standard deviation of $Y$ equal to $0.48\, \sigma _Y$, corresponding to a variance of $0.23\sigma_Y^2$ against the Gibbs value $\sigma_Y^2$ of \cref{sec:lift}. With $\sigma _Y = 1$ and the quartic term included, a Monte Carlo simulation of the no-order dynamics gives a stationary standard deviation of $Y$ of $0.46$. The ambient volatility enters the order-induced responses only through the threshold scale \eqref{eq:dc-vol}, because the price noise cancels in the displacement.

The pool map is \eqref{eq:rho} with $y_\rho = 2$ and floor $\rho_{\mathrm{fl}} = 0.3$. A constant-rate order of size $Q_0$ over $\tau _0$ then lowers the opposing pool to $\rho(1.01) = 0.83$ at the end of execution in the noise-free dynamics. Under the no-order law with $\sigma _Y = 1$, the pool intensity $\rho(Y_t)$ has mean one and a stationary standard deviation of $0.08$. The mean equals one at all times, because the no-order law of $Y_t$ is symmetric and $\rho(y) + \rho(-y) = 2$. The counterflow uses the exponential density with $q_\star = 100$ and no atom, $\pi _0 = 0$. The detection horizon is the reference time, $\tau _d = \tau _0$, and $c_d = 1$, so that \eqref{eq:dc-vol} gives $d_c = \sigma _X\sqrt{ \tau _0} = p_0 = 1$. The threshold scale thus coincides with the price unit by construction rather than by coincidence. The quadratic onset follows from the selected threshold density and linear individual response. The counterflow time scale \eqref{eq:tau-cf} is then $\tau _{\mathrm{cf}} = 2L_0d_c/q_\star = 0.02\,\tau _0$, far below the memory time scales. This separation permits rapid displacement adjustment for suitable order rates and horizons, as discussed in \cref{subsec:pool-bounds}. It does not imply that the pool has equilibrated. The square-root regime occurs at horizons overlapping the memory time scales. The duration grid runs from $0.1\,\tau _0$ to $30\,\tau _0$.

The initial law $\nu _0$ is the finite-start law of \cref{prop:lift}, with $Y_0 = 0$, stationary Gaussian modes $\eta _{i,0}$, $h_{i,0} = - \eta _{i,0}$, and $g_{j,0} = 0$. Symmetry keeps the no-order mean pool intensity equal to one at all times under this law. Initialization can nevertheless affect higher moments, temporal dependence, and the joint law of pool intensity and displacement, and therefore the expected impact. Its standard deviation rises from $0.06$ at $t = \tau _0$ to its stationary value $0.08$ by $t = 5\,\tau _0$.

\begin{table}[!htb]
\centering
\footnotesize
\begin{tabular}{llll}
\hline
Block & Parameter & Baseline value & Basis of choice \\
\hline
Units & $\tau _0,p_0,L_0,Q_0$ & $1,\ 1,\ 1,\ 1$ & Units, $p_0 = \sigma _X\sqrt{ \tau _0}$ \\
Potential & $u_2,\ u_3,\ u_4$ & $1,\ 0,\ 0.1$ & Unit curvature, symmetric \\
Intrinsic memory & $a_i,\ \gamma _i$ & $(0.5, 0.05)$, $(1, 0.1)$ & $S_{YY}^{\mathrm{int}} = 1$, $R_ \gamma = 10$ \\
Cross memory & $c_j^{(0)},\ \lambda _j$ & $(5, 0.5)$, $(2, 0.2)$ & $Q_0$ over $\tau _0$ moves $Y$ by one unit \\
Latent pool & $\rho,\ y_\rho,\ \rho_{\mathrm{fl}}$ & \eqref{eq:rho}, $2$, $0.3$ & $\rho(1.01) = 0.83$ \\
Counterflow & $q_\star,\ c_d,\ \tau _d,\ \pi _0$ & $100,\ 1,\ 1,\ 0$ & $d_c = p_0$, $\tau _{\mathrm{cf}} = 0.02\,\tau _0$ \\
Noise & $\sigma _X,\ \sigma _Y$ & $1,\ 1$ & Price unit, $\rho(Y)$ varies by $8\%$ \\
Initialization & $\nu _0$ & Finite-start law & \Cref{prop:lift}, $Y_0 = 0$ \\
\hline
\end{tabular}
\caption{Baseline model specification for the reported numerical comparisons, informed by parameter-targeting and no-order calculations. All quantities are in the reference units $\tau _0 = p_0 = L_0 = Q_0 = 1$.}
\label{tab:baseline-parameters}
\end{table}

The noise-free cross-memory target gives the unrounded amplitudes $(4.94,0.494)$ for a terminal latent displacement of one unit. The rounded values $(5,0.5)$ used in the table give $1.01$. The linearized stationary covariance is obtained from the continuous Lyapunov equation for the $(Y,h)$ block. The nonlinear no-order Monte Carlo calculation uses a drift-implicit midpoint scheme, a time step of $0.02\tau _0$, $2048$ paths, and a burn-in of $100\tau _0$. These are the settings for the stationary pool statistics, distinct from the $0.01\tau _0$ step used for the reported impact comparisons. Deterministic benchmarks use high-order ODE integration.

\subsection{Verification and numerical error assessment} \label{subsec:test-verification}

Analytical benchmarks and independent numerical comparisons check the expected impact, latent-state covariance, nonlinear counterflow, and execution cost.

First, the Kyle variant without counterflow, $\Pi = 0$, gives a direct test of expected impact. Its baseline-subtracted displacement satisfies
\begin{equation} \label{eq:numerical-constant-liquidity-check}
J_q(H) = \frac{1}{L_0} \int_0^{\min(H,T)} q_s \, ds.
\end{equation}
\Cref{app:pinn} uses it to test the backward solver and the initial-state averaging.

Second, for a quadratic potential the no-order lifted system is linear. Its stationary covariance solves a continuous Lyapunov equation, which gives the linearized variance of $Y$ in \cref{subsec:baseline-parameterization}.

Third, the fresh pool reduces to the ordinary differential equation \eqref{eq:cf-ode}. Its solution provides an independent comparison for nonlinear counterflow during execution and relaxation. The analytical small-order benchmark is checked by reducing the schedule amplitude while holding its shape fixed. Deterministic relaxation formulas, where applicable, and the counterflow bounds provide additional controls. For the stochastic pool, the pathwise bounds \eqref{eq:pool-bounds} and \eqref{eq:pool-relaxation-bounds} must hold on every simulated path, with the realized range of $\rho(Y_t)$.

Fourth, the one-sided form of the cost identity \eqref{eq:roundtrip-identity} provides a check on the execution-weighted displacement of \cref{sec:schedules}.

\myparagraph{Numerical accuracy.}
The finite-exponential lift is exact for the specified kernels. Changing the
modes changes the model unless a separate target kernel is being approximated.
Monte Carlo estimates retain sampling and time-discretization errors. Paired
standard errors account for covariance in differences and ratios at a fixed
time step. Time-step sensitivity was checked at $dt / \tau _0 = 0.01$,
$0.005$, $0.0025$, and $0.00125$ for an isolated flat order, the interrupted
schedule, and an inherited-state probe, all using the baseline GLE pool.
For the isolated and interrupted orders, both with $Q = Q_0$ and
$T = \tau _0$, changes in terminal response between successive settings
are below the reported paired standard errors of the terminal estimates.
The runs reuse a random seed but do not couple Brownian increments across
time steps, so these differences contain sampling variation as well as
discretization effects.

The inherited-state test uses a flat prior order of size
$Q_{\mathrm{prior}} = 5Q_0$, followed at zero gap by a flat probe of size
$Q_0$, each of duration $\tau _0$. The ratio of mean incremental probe impact
with and without the prior order, minus one, changes by
$4.2\times10^{-4}$ ($0.042$ percentage points) between
$dt = 0.01\,\tau _0$ and $dt = 0.0025\,\tau _0$. The change between
$0.0025\,\tau _0$ and $0.00125\,\tau _0$ is $7\times10^{-6}$
($0.0007$ percentage points). This zero-gap refinement check does not confirm time-step convergence of the long-gap results.

\section{Numerical experiments} \label{sec:experiments}

The numerical results compare the size and duration dependence of impact, the effect of execution timing on depletion and recovery, and the response to an order placed after earlier trading. The baseline includes rapid counterflow, a positive pool floor, and separate intrinsic and order-flow memory spectra. The common computational protocol and verification procedures are given in \cref{subsec:computational-protocol,subsec:test-verification}.

The baseline is symmetric under simultaneous reversal of the order, displacement, latent state, and memory variables. The figures therefore display buy orders. Sell orders provide a symmetry check for this specification. Asymmetric potentials or state-dependent cross-memory require a comparison of both signs.

\begin{table}[!htb]
\centering
\small
\begin{tabular}{p{0.23\textwidth}p{0.68\textwidth}}
\hline
Comparison & Settings and observables \\
\hline
Impact and concavity & Four variants on the $81$-point size grid, with $T / \tau _0 \in \{1,3,10\}$ in the figure; terminal impact and local exponents. \\
Schedule and recovery & Four GLE schedules with common size and duration; execution and post-execution displacement, pool intensity, counterflow, and normalized recovery. \\
Inherited state & Unit-size probe and unit-duration orders; the baseline uses a prior size $Q_0$, while the additional comparison varies prior size over $Q_0$, $3Q_0$, and $5Q_0$ for three spectra and nine gaps. The fresh-pool reference is retained for the baseline decomposition. \\
\hline
\end{tabular}
\caption{Organization of the numerical comparisons. The short-gap inherited-state plot is distinguished from the extended comparisons across spectra and prior-order sizes.}
\label{tab:compact-study}

\end{table}

The stochastic results in the figures use Monte Carlo estimates with $2048$ paths and a time step of $0.01 \tau _0$ at the baseline of \cref{tab:baseline-parameters}.

\subsection{Impact, concavity, and execution duration} \label{subsec:test-concavity}

The first comparison uses constant-rate orders,
\begin{equation} \label{eq:const-rate}
q_t(Q,T) = \frac{Q}{T} \ind_{[0,T]}(t).
\end{equation}
The Kyle benchmark removes counterflow, the fresh pool removes depletion, and the single-mode and GLE pools retain depletion with different relaxation spectra. Their terminal and peak expected responses are
\begin{align} \label{eq:signed-response}
I_\varepsilon(V,T) &= \varepsilon J_{q(\varepsilon V,T)}(T), \qquad \varepsilon \in \{ - 1,1\}, \\
I_\varepsilon^{\max}(V,T) &= \max_{0 \le t \le T} \varepsilon J_{q(\varepsilon V,T)}(t). \nonumber
\end{align}
For the buy orders displayed below, the subscript $\varepsilon = 1$ is suppressed. The peak is the maximum of the expected response, rather than the expectation of a pathwise maximum.

For positive terminal impact, the local exponent and its centered approximation on a logarithmic size grid are
\begin{align} \label{eq:numerical-local-exponent}
\delta(V,T) &= \frac{\partial \log I(V,T)}{\partial \log V}, \\
\widehat\delta(V_k,T) &= \frac{\log I(V_{k + 1},T) - \log I(V_{k - 1},T)}{\log V_{k + 1} - \log V_{k - 1}}. \nonumber
\end{align}
Concavity is assessed separately through the marginal slopes,
\begin{equation} \label{eq:impact-discrete-slope}
s_k(T) = \frac{I(V_{k + 1},T) - I(V_k,T)}{V_{k + 1} - V_k},
\end{equation}
which must be nonincreasing within numerical uncertainty. A logarithmic slope below one does not by itself imply concavity.

\begin{figure}[!htb]
\centering
\IfFileExists{figure_1_impact.pdf}{\includegraphics[width=\textwidth]{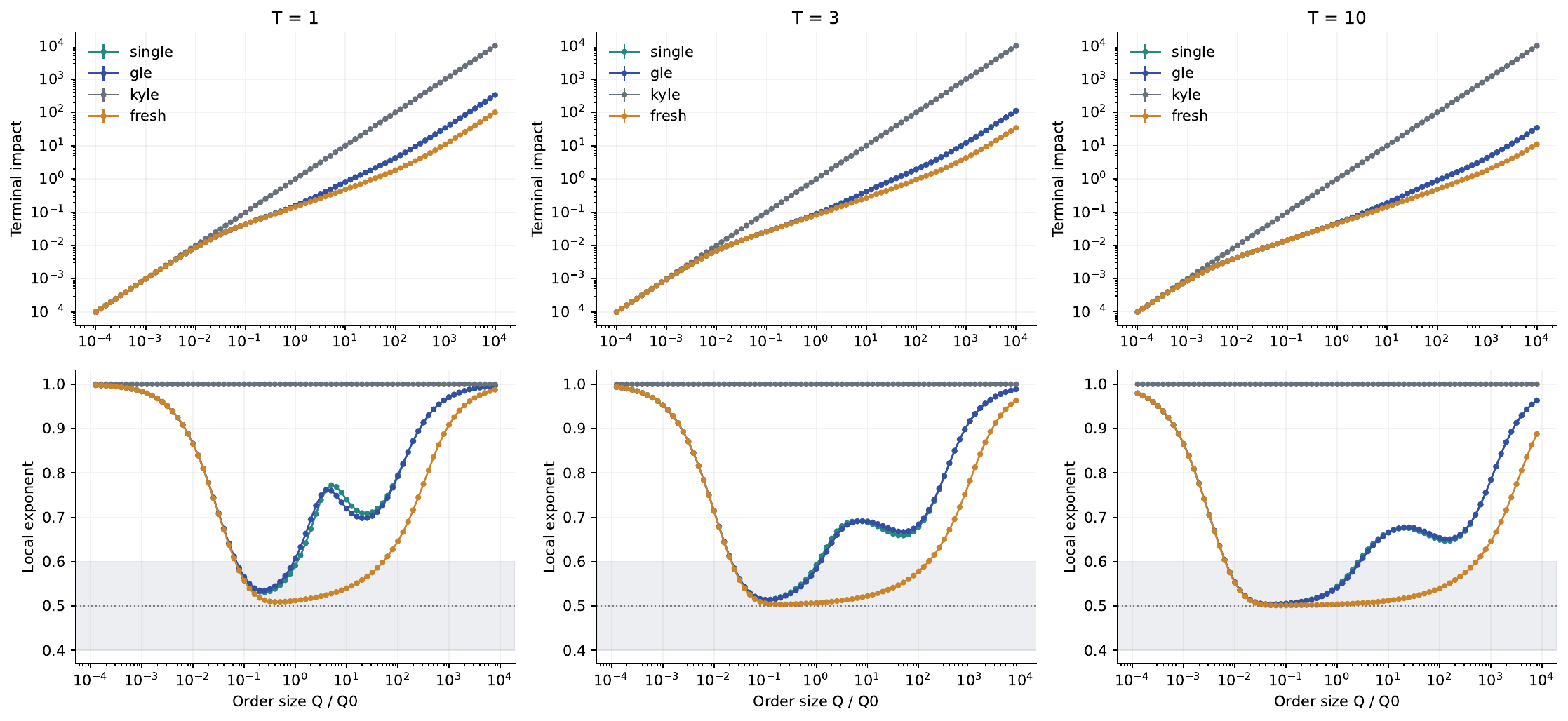}}{\fbox{\parbox[c][0.22\textheight][c]{0.92\textwidth}{\centering Figure placeholder: \texttt{figure\_1\_impact.pdf}}}}
\caption{Terminal impact and local size exponent for the Kyle, fresh-pool, single-mode, and GLE variants on $81$ logarithmically spaced sizes from $10^{-4}Q_0$ to $10^{4}Q_0$. The columns correspond to $T/\tau_0 = 1$, $3$, and $10$. Upper panels show terminal impact. Lower panels show the centered logarithmic slope $\widehat\delta$ of \eqref{eq:numerical-local-exponent}. The dotted line marks the square-root exponent $1/2$, and the shaded band is $[0.4,0.6]$. The fresh-pool curve solves \eqref{eq:cf-ode}, and the Kyle line is analytical. The single-mode and GLE curves are Monte Carlo estimates with $2048$ paths and a time step of $0.01\,\tau_0$.}
\label{fig:impact}
\end{figure}

\Cref{fig:impact} separates the Kyle line from the three variants with counterflow. The Kyle response is linear, $I(V,T) = V / L_0$, and independent of duration.
The fresh-pool response is lower and grows sublinearly over the displayed range. Its centered logarithmic slope at $Q_0$ from the dense grid is $0.51$ at $T = \tau _0$.
Increasing duration lowers the fresh-pool impact at a fixed size, consistently with the additional time available for counterflow to absorb the order.

Below about $0.1\,Q_0$ the single-mode and GLE impacts lie close to each other and to the fresh-pool response, because a small order barely depletes the pool. At larger sizes they rise above it. For $Q = 10Q_0$ and $T = 0.5 \tau _0$ the GLE and single-mode impacts are $1.23 p_0$ and $1.13 p_0$, a relative difference of approximately $9\%$, against $0.71 p_0$ for the fresh pool. For $T = \tau _0$ they are $0.82 p_0$ and $0.80 p_0$, against $0.48 p_0$. Depletion weakens the counterflow for larger orders and steepens the curve. On the dense grid the depleted-pool exponent has a local maximum, about $0.76$ near $4\,Q_0$ at $T = \tau_0$, where the pool intensity falls fastest with order size. It falls back to about $0.70$ near $20\,Q_0$, where the intensity is close to its floor, before the linear large-order limit takes over. The marginal slopes remain nonincreasing throughout, so impact stays concave.

Without the floor, depletion steepens the curve instead. In noise-free
calculations with $\rho_{\mathrm{fl}} = 0$, the largest local exponent grows
with the horizon, from about $1.86$ at $T = \tau_0$ through $2.53$ near
$5\,\tau_0$ and $2.79$ at $10\,\tau_0$ to about $3.07$ at $20\,\tau_0$.
With $\rho_{\mathrm{fl}} = 0.3$ it stays at or below $1.0$ for
$T \ge 0.3\,\tau_0$.

For the stochastic baseline, the centered slopes at $Q_0$ from the dense grid are $0.59$ and $0.61$ at $T = \tau_0$ for the single-mode and GLE pools, against $0.51$ for the fresh pool. The single-mode value lies inside the stated band $[0.4,0.6]$, while the GLE value lies just above it.

At $Q = Q_0$ and $T = \tau _0$ the GLE impact is $0.158 p_0$, below the pathwise upper bound of $0.269 p_0$ that \cref{prop:pool-bounds} gives with $\rho_{\min} = \rho_{\mathrm{fl}}$. The six-duration comparison at $Q = Q_0$ in \cref{tab:impact-benchmark} gives relative GLE--single differences of at most
$1.6\%$ in magnitude, with paired standard errors of $0.01\%$--$0.13\%$. The
positive differences at short durations and the sign change between $T = \tau _0$
and $T = 3\tau _0$ are resolved above the paired sampling noise. The $T = 30\tau
_0$ difference of $- 0.17\%$ is comparable to its standard error of $0.12\%$ and
is not resolved at the reported path count. This limited sensitivity at unit
order size does not extend uniformly to larger orders.

\begin{table}[!htb]
\centering
\small
\begin{tabular}{lrrrr}
\hline
$T/\tau _0$ & $I_T^{\mathrm{fresh}}/p_0$ & $I_T^{\mathrm{single}}/p_0$ & $I_T^{\mathrm{GLE}}/p_0$ & Difference (\%) \\
\hline
0.1 & $0.465086$ & $0.468880$ $(3.6\times10^{-5})$ & $0.471458$ $(4.7\times10^{-5})$ & $+ 0.55$ $(0.01)$ \\
0.3 & $0.269390$ & $0.277172$ $(7.1\times10^{-5})$ & $0.281566$ $(9.2\times10^{-5})$ & $+ 1.59$ $(0.04)$ \\
1 & $0.144834$ & $0.155852$ $(9.7\times10^{-5})$ & $0.158264$ $(1.1\times10^{-4})$ & $+ 1.55$ $(0.10)$ \\
3 & $0.082776$ & $0.090183$ $(7.0\times10^{-5})$ & $0.089447$ $(7.8\times10^{-5})$ & $- 0.82$ $(0.12)$ \\
10 & $0.045057$ & $0.046978$ $(4.0\times10^{-5})$ & $0.046855$ $(4.3\times10^{-5})$ & $- 0.26$ $(0.13)$ \\
30 & $0.025932$ & $0.026375$ $(2.1\times10^{-5})$ & $0.026330$ $(2.3\times10^{-5})$ & $- 0.17$ $(0.12)$ \\
\hline
\end{tabular}
\caption{Terminal impact at $Q = Q_0$ for the fresh, single-mode, and GLE
pools. The fresh-pool values come from the deterministic ODE
\eqref{eq:cf-ode} and carry no sampling error. The single-mode and GLE
values are Monte Carlo estimates with $2048$ paths and a time step of
$0.01\tau _0$. Their paired standard errors are given in parentheses. The
final column is $100(I_T^{\mathrm{GLE}}/I_T^{\mathrm{single}} - 1)$, with
its paired delta-method standard error in parentheses. The two Monte Carlo
variants share the same random innovations, so the ratio SE is a paired
estimate. The analytical Kyle benchmark is $I_T/p_0 = 1$ at every duration.}
\label{tab:impact-benchmark}

\end{table}

The dense grid also resolves the dependence on duration. The size--duration specification covers $10^{ - 4} Q_0 \le V \le 10^4 Q_0$ and $T / \tau _0 \in \{0.1,0.3,1,3,10,30\}$. The band estimates below use the same $81$ logarithmically spaced sizes as the dense comparison. An approximately square-root interval $[V_-,V_+]$ satisfies
\begin{equation} \label{eq:powerlaw-operational-criterion}
\left| \delta(V_k,T) - \frac{1}{2} \right| \le \epsilon_{\mathrm{sq}}, \qquad \epsilon_{\mathrm{sq}} = 0.1,
\end{equation}
throughout a connected range with positive increasing impact and nonincreasing marginal slopes. Its width and mean exponent are
\begin{equation} \label{eq:powerlaw-width}
W(T) = \log_{10}\frac{V_+}{V_-}, \qquad \bar\delta(T) = \frac{1}{N_{\mathrm{band}}} \sum_{V_k \in [V_-,V_+]} \delta(V_k,T).
\end{equation}
The reported band criterion additionally requires $W(T) \ge 1$, with $W$ measured in decades, that is, in units of $\log_{10}(V_+/V_-)$. A shorter interval is recorded as a crossover. Band boundaries are located on the stated size grid.

\begin{table}[!htb]
\centering
\small
\begin{tabular}{lrrrrrr}
\hline
$T / \tau _0$ & $0.1$ & $0.3$ & $1$ & $3$ & $10$ & $30$ \\
\hline
Fresh-pool width  & -- & $1.70$ & $2.80$ & $3.70$ & $4.80$ & $5.70$ \\
Single-mode width & -- & --     & $1.10$ & $1.60$ & $2.50$ & $3.40$ \\
GLE width         & -- & --     & $1.00$ & $1.60$ & $2.50$ & $3.40$ \\
\hline
\end{tabular}
\caption{Widths of intervals satisfying the square-root exponent tolerance, $\epsilon_{\mathrm{sq}} = 0.1$, at the baseline, in decades. All three rows use the same $81$-point size grid on $10^{-4} \le V/Q_0 \le 10^4$ and the same six durations. The fresh row is from the deterministic ODE \eqref{eq:cf-ode}. The single-mode and GLE rows are Monte Carlo estimates with $2048$ paths and a time step of $0.01\,\tau_0$. Entries marked `--` are cases where no connected interval satisfied both the exponent tolerance and the one-decade minimum-width criterion. They are not missing calculations.}
\label{tab:impact-bands}

\end{table}

The widths in \cref{tab:impact-bands} increase with execution duration. For the fresh pool the one-decade criterion is first met at $T = 0.3\,\tau_0$, where the band spans $1.70$ decades. Between the grid durations, it is first met near $T = 0.15\,\tau_0$. The depleted pools first meet it at $T = \tau_0$, with bands of $1.00$ to $1.10$ decades. Depletion narrows the interval by 1.7 to 1.8 decades at $T = \tau_0$ and by 2.3 decades at $T = 30\,\tau_0$. For the duration-free thresholds of \cref{rem:duration-free} with $c_\tau = 1$, the fresh-pool band spans $2.80$ decades for $\tau_d = c_\tau T$, the fixed-horizon value at this duration, and $2.70$ decades under elapsed-time scaling at $T = \tau_0$. At $T = 10\,\tau_0$ both span $3.80$ decades, against $4.80$ for the fixed horizon. The single-mode and GLE widths are $1.10$ and $1.00$ decades at $T = \tau_0$ and coincide at the reported precision for the longer durations. These observations suggest that the counterflow determines the principal concave response, while depletion limits its range. The similar widths indicate limited sensitivity of this grid statistic to the two matched spectra.

The $81$-point size grid has spacing $0.1$ decade. Cubic interpolation of
these same mean impacts in logarithmic size and impact coordinates onto
$241$ logarithmically spaced sizes gives widths of $1.17$ and $1.67$
decades for the single-mode pool at $T = \tau _0$ and $T = 3\tau _0$, and
$1.10$ and $1.70$ decades for the GLE pool at the same durations. All four
widths remain above one decade under this interpolation. The single-mode
width exceeds the GLE width by about $0.07$ decade at $T = \tau _0$. At
$T = 3\tau _0$ they differ by one interpolated interval, about $0.033$
decade. No additional order sizes were simulated, so this check shows
sensitivity to interpolation rather than convergence under size-grid
refinement.

A separate $500$-draw sensitivity check perturbs each mean independently
by a zero-mean Gaussian variable with standard deviation equal to its reported
standard error and recomputes the bands on the original $81$-point grid.
The conditional width summaries remain unchanged at the reported precision.
The procedure omits draws with no qualifying band and ignores correlations
across order sizes induced by the common random numbers, so it does not
prove robustness of the one-decade classification to sampling uncertainty.

\subsection{Schedule dependence and post-execution recovery} \label{sec:schedules}

The second comparison holds order size and elapsed execution duration fixed and changes the timing of execution. The reference setting is $Q = Q_0$ and $T = \tau _0$. Writing $q_t = Q \psi_T(t)$ with $\int_0^T \psi_T(s) \, ds = 1$, the flat, front-loaded, and back-loaded profiles are
\begin{equation} \label{eq:schedule-comparison}
\psi_T^{\mathrm{flat}}(s) = \frac{1}{T}, \qquad \psi_T^{\mathrm{front}}(s) = \frac{2(T - s)}{T^2}, \qquad \psi_T^{\mathrm{back}}(s) = \frac{2s}{T^2}, \qquad 0 \le s \le T.
\end{equation}
The interrupted schedule contains a symmetric pause with fraction $\kappa = 0.3$,
\begin{equation} \label{eq:pause-profile}
\psi_{T,\kappa}^{\mathrm{pause}}(s) = \frac{\ind_{[0,(1 - \kappa)T/2]}(s) + \ind_{[(1 + \kappa)T/2,T]}(s)}{(1 - \kappa)T}.
\end{equation}
The pause therefore occupies $[0.35T,0.65T]$. Since elapsed duration and size are fixed, the active rate exceeds that of the flat schedule. This comparison includes both the interruption and the compensating increase in the active rate.

\begin{table}[!htb]
\centering
\small
\begin{tabular}{lccccccccc}
\hline
Schedule & $J_q(T)$ & $\mathrm{SE}$ & Peak $J_q$ & $C[q]$ & $\mathrm{SE}$ & Cum.\ cf & $\mathrm{SE}$ & $\tau_{0.5}$ & $\tau_{0.1}$ \\
\hline
flat & $0.1581$ & $1.1\times10^{-4}$ & $0.1581$ & $0.1352$ & $5.7\times10^{-5}$ & $0.8419$ & $1.1\times10^{-4}$ & $0.161$ & $1.389$ \\
front & $0.0792$ & $7.0\times10^{-5}$ & $0.1835$ & $0.1487$ & $4.5\times10^{-5}$ & $0.9209$ & $7.0\times10^{-5}$ & $0.319$ & $2.640$ \\
back & $0.2152$ & $1.5\times10^{-4}$ & $0.2152$ & $0.1546$ & $8.0\times10^{-5}$ & $0.7848$ & $1.5\times10^{-4}$ & $0.117$ & $1.043$ \\
interrupted & $0.1875$ & $1.3\times10^{-4}$ & $0.1875$ & $0.1441$ & $5.2\times10^{-5}$ & $0.8125$ & $1.3\times10^{-4}$ & $0.135$ & $1.179$ \\
\hline
\end{tabular}
\caption{Schedule comparison and recovery in the GLE pool at $Q = Q_0$, $T =
\tau_0$, with $2048$ paths and a time step of $0.01\,\tau_0$. Terminal and peak
impacts are in units of $p_0$. The SE columns give paired standard errors of the
per-path order-minus-control differences for the three quantities that are sample
means: terminal impact $J_q(T)$, execution-weighted displacement $C[q]$, and
cumulative counterflow. The peak and the relaxation times are a maximum and
crossing times of the mean trajectory, not sample means, so they do not carry
standard errors of the same kind. Relaxation times are computed on a non-uniform
observation grid with spacing $0.005\,\tau_0$ on $[T, T + 2\tau_0]$, so they are
resolved to within linear-interpolation error, except the front-loaded
$\tau_{0.1}$, which lies beyond this window.}
\label{tab:schedule-summary}

\end{table}

\begin{figure}[!htb]
\centering
\IfFileExists{figure_2_schedules.pdf}{\includegraphics[width=0.8\textwidth]{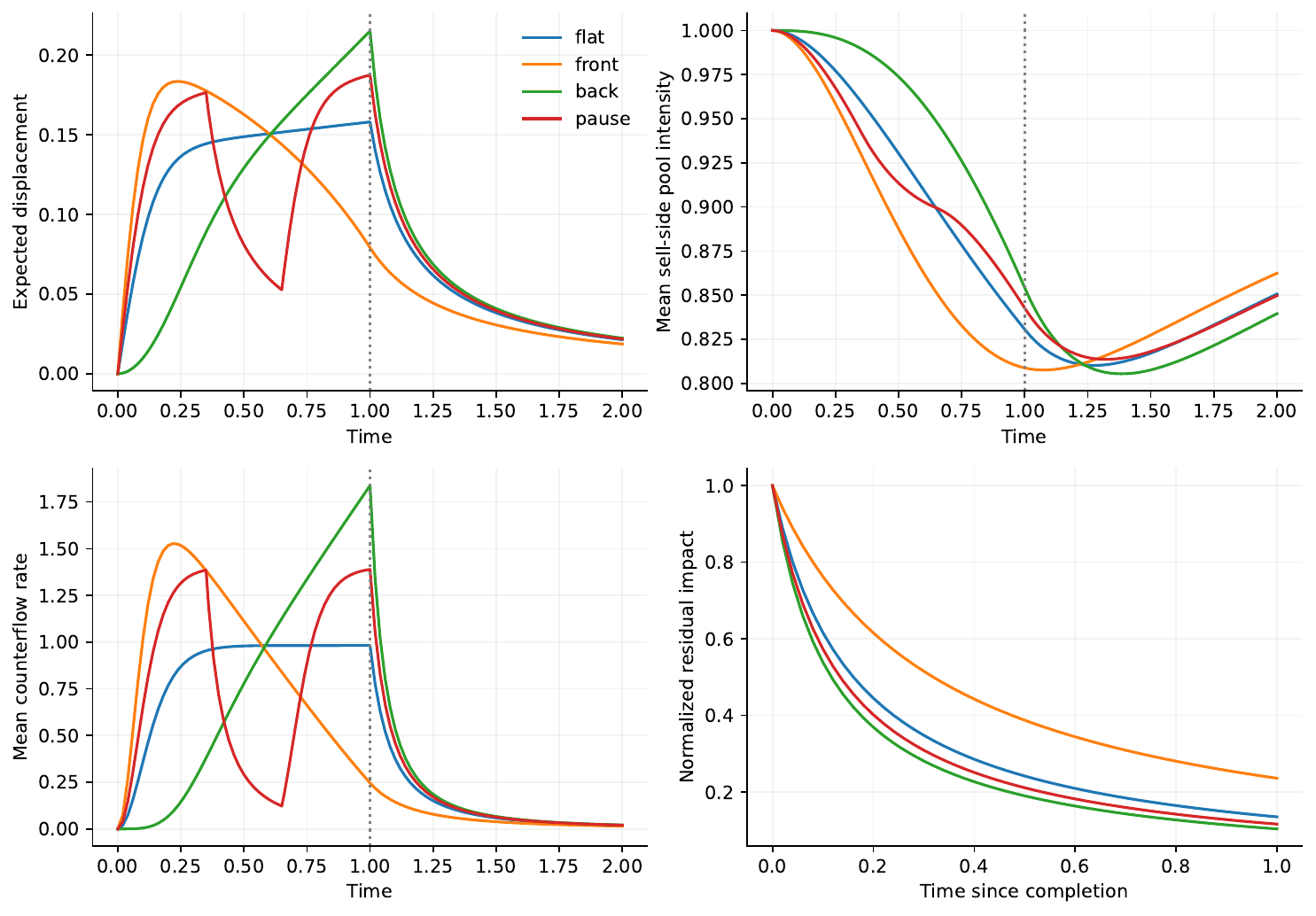}}{\fbox{\parbox[c][0.22\textheight][c]{0.92\textwidth}{\centering Figure placeholder: \texttt{figure\_2\_schedules.pdf}}}}
\caption{Schedule dependence and recovery in the GLE pool for flat, front-loaded, back-loaded, and interrupted buy schedules with common size and duration, at $Q = Q_0$ and $T = \tau_0$. The panels display expected displacement (upper left), mean sell-side pool intensity (upper right), the reported mean counterflow rate (lower left), and normalized residual impact $R_q(\tau)$ after completion (lower right). The vertical dotted line marks completion at $t / \tau_0 = 1$. The interrupted schedule pauses between $0.35T$ and $0.65T$. All panels show the first $\tau_0$ after completion, so the values of $\tau_{0.1}$ in \cref{tab:schedule-summary} lie beyond the plotted window. Pool depletion continues briefly after completion, while displacement begins to decline. The simulated displacement and counterflow satisfy the volume balance implied by \eqref{eq:displacement-ode}.}
\label{fig:schedules}
\end{figure}

The displacement panel in \cref{fig:schedules} shows that front-loading produces the largest response early in execution, followed by a decline as the trading rate falls. Its peak is $0.1835 p_0$ and it ends at $0.0792 p_0$, the lowest terminal impact of the four schedules. Back-loading gives the smallest early response and the largest terminal response, $0.2152 p_0$. The flat schedule ends at $0.1581 p_0$, and the interrupted schedule at $0.1875 p_0$, higher than flat because its active rate is larger. The flat run is separate from the one behind \cref{tab:impact-benchmark}, and the two terminal values differ by about one standard error of their difference. The paired standard errors of these terminal values are between $7.0\times10^{-5}$ and $1.5\times10^{-4}$, so the ordering across schedules is resolved well above the paired sampling noise at $2048$ paths.

The pool panel records a delayed response to trading. Front-loading brings depletion forward, whereas back-loading postpones it. The interruption slows depletion but does not immediately restore the pool. For all four schedules the mean pool intensity continues to fall for some time after completion and then recovers. Its minimum lies between $0.806$ and $0.814$ and is reached between $t = 1.08 \tau _0$ and $1.38 \tau _0$. This lag is consistent with the persistence of the order-flow and intrinsic-memory states after the current execution rate has fallen to zero. The mean intensity is still below its initial level at the end of the plotted window.

Terminal impact alone does not measure execution cost. The additional statistic
\begin{equation} \label{eq:execution-displacement}
C[q] = \frac{1}{Q} \int_0^T q_s J_q(s) \, ds, \qquad Q > 0,
\end{equation}
weights displacement by the executed volume. It excludes spreads, fees, and fill-price effects. The values of $C[q]$ in \cref{tab:schedule-summary} do not follow the terminal-impact ranking. Front-loading has the lowest terminal impact, $0.0792 p_0$, but the second-highest execution-weighted displacement, $0.1487$. It executes three quarters of its volume in the first half of the horizon, where its displacement is largest. The flat schedule has the lowest execution-weighted displacement, $0.1352$, although its terminal impact is twice that of front-loading. Integrating \eqref{eq:displacement-ode} gives the one-sided form of \eqref{eq:roundtrip-identity}, $Q\,C[q] = \tfrac{L_0}{2}\EE[D_T^2] + \EE\int_0^T q_t^{\mathrm{cf}}D_t\,dt$, which provides a check on the tabulated values. The paired standard errors of $C[q]$ are $4.5\times10^{-5}$ to $8.0\times10^{-5}$, so the differences among schedules are resolved at the reported path count.

\myparagraph{Recovery after execution.} \label{sec:relaxation}
After completion, $q_t = 0$ and the order-flow memory decays according to
\begin{equation} \label{eq:gj-decay}
g_{j,t} = g_{j,T} e^{ - \lambda_j(t - T)}.
\end{equation}

For a nonzero terminal response, normalized residual impact is
\begin{equation} \label{eq:normalized-relaxation}
R_q(\tau) = \frac{J_q(T + \tau)}{J_q(T)}, \qquad \tau \ge 0.
\end{equation}
The lower-right panel shows rapid normalized decay for all four schedules. The crossing time
\begin{equation} \label{eq:relaxation-crossing-time}
\tau_\alpha = \inf\{\tau \ge 0 : R_q(\tau) \le \alpha\}, \qquad \alpha \in \{0.5,0.1\},
\end{equation}
is reported in \cref{tab:schedule-summary}. It is computed on a non-uniform observation grid with spacing $0.005\,\tau_0$ on $[T, T + 2\tau_0]$.

Residual displacement at a finite horizon is not evidence of permanent impact. Without counterflow, the expected response remains constant after execution. For a fresh pool and positive terminal displacement, the independent relaxation benchmark is
\begin{equation} \label{eq:exact-counterflow-relaxation}
\tau = L_0 \int_{D_{T + \tau}}^{D_T} \frac{dD}{\mathcal A(D)}.
\end{equation}
The stochastic-pool bounds in \cref{prop:pool-bounds} provide the corresponding controls when the pool evolves.

The counterflow panel is checked against the displacement. In the stated units, integration of \eqref{eq:displacement-ode} gives the exact relation
\begin{equation} \label{eq:schedule-volume-balance}
L_0 J_q(t) + \int_0^t \EE\left[q_s^{\mathrm{cf},[q]} - q_s^{\mathrm{cf},[0]}\right] \, ds = \int_0^{\min(t,T)} q_s \, ds.
\end{equation}
For an isolated order started at zero displacement, the no-order counterflow is zero. The simulated displacement and accumulated counterflow satisfy this identity to within $3 \times 10^{-10}$ for all four schedules, and the plotted counterflow rate integrates to the accumulated counterflow within about $0.1\%$. The mean counterflow rate of the flat schedule levels off at $0.98 Q_0/\tau _0$, just below the order rate. By completion, counterflow has absorbed $84.2\%$ of the flat order, since $L_0 J_q(\tau_0) = 0.1581 Q_0$. Front-loading absorbs $92.1\%$ and back-loading $78.5\%$.

\subsection{Inherited-state effects and the memory spectrum} \label{sec:inherited}

The third comparison measures how a prior order changes the impact of a subsequent probe. The baseline comparison uses two buy orders of size $Q_0$ and duration $\tau _0$; \cref{sec:inherited-depletion} extends it to larger prior orders. The gap $G$ is measured between completion of the prior order and arrival of the probe, whose start time is
\begin{equation} \label{eq:compact-gaps}
t_G = \tau _0 + G.
\end{equation}
The current figure contains three short-gap cases between $G = 0$ and $G = \tau _0$. The extended comparison summarized below reaches $G = 50 \tau _0$ and addresses the longer tail of the inherited state. The comparison retains the full state at probe arrival, including residual displacement and every latent-memory variable.

Two histories evolve from the same initial law: one includes the prior order, and the other does not. At $t_G$, each history branches into a continuation with the probe and a continuation without it. Let $D^{[1,1]}$, $D^{[1,0]}$, $D^{[0,1]}$, and $D^{[0,0]}$ denote the displacements when both orders, only the prior order, only the probe, or neither order is present. The incremental probe impacts are
\begin{align} \label{eq:conditional-second-impact}
J_{2|1}(G) &= \EE\left[D_{t_G + \tau _0}^{[1,1]} - D_{t_G + \tau _0}^{[1,0]}\right], \\
J_{2|0}(G) &= \EE\left[D_{t_G + \tau _0}^{[0,1]} - D_{t_G + \tau _0}^{[0,0]}\right]. \nonumber
\end{align}
The histories are evaluated at the same elapsed time from initialization. This removes the change in the finite-start law from the comparison. For simulation estimates, common random numbers preserve the pairing across histories and continuations, while future innovations are independent of the prehistory.

The relative history effect is
\begin{equation} \label{eq:prior-history-effect}
\mathcal H(G) = \frac{J_{2|1}(G)}{J_{2|0}(G)} - 1,
\end{equation}
provided the denominator is numerically resolved. A prior buy order has two competing consequences. Its remaining displacement raises the counterflow faced by the probe and can reduce incremental impact. Its depleted opposing pool weakens counterflow and can increase that impact. The fresh-pool comparison retains the residual-displacement effect without depletion. The excess history effect relative to that reference is
\begin{equation} \label{eq:pool-channel}
\mathcal H_{\mathrm{pool}}(G) = \mathcal H(G) - \mathcal H_{\mathrm{fresh}}(G).
\end{equation}
This is an operational comparison across models. It includes the interaction of depletion with the inherited displacement and is not an exact additive separation of independent mechanisms.

Baseline calculations give $\mathcal H_{\mathrm{fresh}}(0) \simeq - 12.3\%$ and $\mathcal H_{\mathrm{pool}}(0) \simeq 7.5\%$ for the GLE pool. The two contributions give a net reduction of about $5\%$ in the incremental probe impact. Thus a negative total history effect does not imply that prior trading leaves the pool unaffected. In these calculations the additional impact associated with depletion partly offsets the reduction caused by residual displacement. At $G = 50 \tau _0$, the fresh-pool effect is still approximately $- 0.27\%$, consistent with its slowly decaying post-execution displacement.

\myparagraph{Spectral comparison.} \label{sec:powerlaws}
The baseline GLE and single-mode pools are compared with a broader specification containing four modes in each kernel and spectral breadth $1000$. The breadth measures are
\begin{equation} \label{eq:Rlambda}
R_\gamma = \frac{\gamma_{\max}}{\gamma_{\min}}, \qquad R_\lambda = \frac{\lambda_{\max}}{\lambda_{\min}}.
\end{equation}
Rates are logarithmically spaced around the corresponding baseline geometric rate. Integrated strengths are preserved,
\begin{align} \label{eq:kernel-strength-controls}
S_{YY}^{\mathrm{int}} &= \sum_{i = 1}^N \frac{a_i}{\gamma_i}, \qquad S_{YY}^{0} = \sum_{i = 1}^N a_i, \\
S_{YX}^{\mathrm{int}} &= \sum_{j = 1}^M \frac{c_j^{(0)}}{\lambda_j}, \qquad S_{YX}^{0} = \sum_{j = 1}^M c_j^{(0)}. \nonumber
\end{align}
The zero-lag strengths are distinct quantities and need not remain fixed. In particular,
\begin{equation} \label{eq:spectrum-force-variance}
\EE[\xi_t^2] = \sigma_Y^2 K_{YY}(0),
\end{equation}
so changing the intrinsic spectrum at fixed integrated strength and fixed $\sigma_Y$ can also change the force variance. The comparison measures the effect of these specified spectra under a common integrated-strength convention.

\begin{table}[!htb]
\centering
\small
\begin{tabular}{lcccc}
\hline
Pool specification & $\mathcal H_{\mathrm{pool}}(0)$ & $\mathrm{SE}$ & Half-decay time$/\tau _0$ & $\mathcal H_{\mathrm{pool}}(50 \tau _0)$ \\
\hline
Single mode & $9.3\%$ & $0.02\%$ & $2.15$ & $\simeq 0.02\%$ \\
Baseline GLE & $7.5\%$ & $0.03\%$ & $2.05$ & $\simeq 0.02\%$ \\
Broad spectrum & $6.2\%$ & $0.03\%$ & $1.74$ & $\simeq 0.10\%$ \\
\hline
\end{tabular}
\caption{Inherited-state comparison at common integrated kernel strengths. The
half-decay time refers to $\mathcal H_{\mathrm{pool}}(G)$ relative to its value
at zero gap, log-linearly interpolated between grid points. Percentages refer to
the relative probe-impact difference in \eqref{eq:pool-channel}; $0.10\%$
corresponds to $0.001$ in that equation. Paired delta-method standard errors of
$\mathcal H_{\mathrm{pool}}(0)$ are at most $0.03$ percentage points, so the
values of $\mathcal H_{\mathrm{pool}}(0)$ are resolved above the paired sampling
noise at $2048$ paths.}
\label{tab:memory-history}

\end{table}

The extended summary in \cref{tab:memory-history} shows that a broader spectrum need not slow every stage of recovery. Its relative pool effect halves sooner, but its long-gap residual is larger. The paired standard errors are at most $0.03$ percentage points, so the ordering across pools is resolved above the paired sampling noise. Fast modes can reduce the initial response while slow modes sustain a smaller tail. This distinction explains why a single half-decay time is insufficient to characterize recovery. The initial amplitude, intermediate decay, and long-gap residual describe different parts of the response. Together they distinguish the response profiles beyond the single-order band width.

The inherited displacement accounts for the fresh-pool effect quantitatively. The counterflow returns the displacement to the same quasi-steady path wherever the probe starts, so the probe forgets its initial displacement well before completion. Its increment is therefore reduced by the residual of the prior order at probe completion, and $\mathcal H_{\mathrm{fresh}}(G) \simeq - D_{t_G + \tau _0}^{[1,0]}/J_{2|0}(G)$. The simulated residuals at $t_G + \tau _0$ are $0.0178 p_0$ for $G = 0$ and $0.0094 p_0$ for $G = \tau _0$, which give $- 12.28\%$ and $- 6.51\%$, as measured. At long gaps the residual follows the large-$\tau$ limit $L_0/(\omega \tau)$ of \eqref{eq:quadratic-relaxation}, with $\omega = q_\star/(2 d_c^2) = 50$ at the baseline. At the probe start it equals $0.0020$, $0.0010$, and $0.0004 p_0$ for $G = 10$, $20$, and $50 \tau _0$, as this limit predicts, and the approximation gives $- 0.27\%$ at $G = 50 \tau _0$. The two contributions therefore decay at different rates. The residual-displacement contribution decays as $1/G$, whereas the pool-associated difference in \cref{tab:memory-history} halves within about $2 \tau _0$. The total effect is positive where the latter exceeds the residual-displacement contribution in magnitude. For the single-mode pool this happens at $G = \tau _0$ and $2 \tau _0$, where the total is approximately $+ 0.6\%$ and $+ 0.5\%$, and the total is negative again from $G = 5 \tau _0$. The GLE total remains negative at every gap and approaches the fresh-pool value, $- 0.25\%$ against $- 0.27\%$ at $G = 50 \tau _0$. The broad-spectrum total, $- 0.17\%$ at that gap, still carries its heavier pool tail. The sign of the total effect alone therefore says nothing about depletion.

The corresponding initial-pool diagnostic is
\begin{equation} \label{eq:inherited-pool-difference}
\Delta\rho_{\mathrm{init}}(G) = \EE\left[\rho(Y_{t_G}^{[1]}) - \rho(Y_{t_G}^{[0]})\right],
\end{equation}
where the superscript records whether the prior order is present. Together, this quantity, the total history effect, and the fresh-pool comparison distinguish the observed state difference from its effect on the incremental price response.

\begin{figure}[!htb]
\centering
\IfFileExists{figure_3_memory.pdf}{\includegraphics[width=\textwidth]{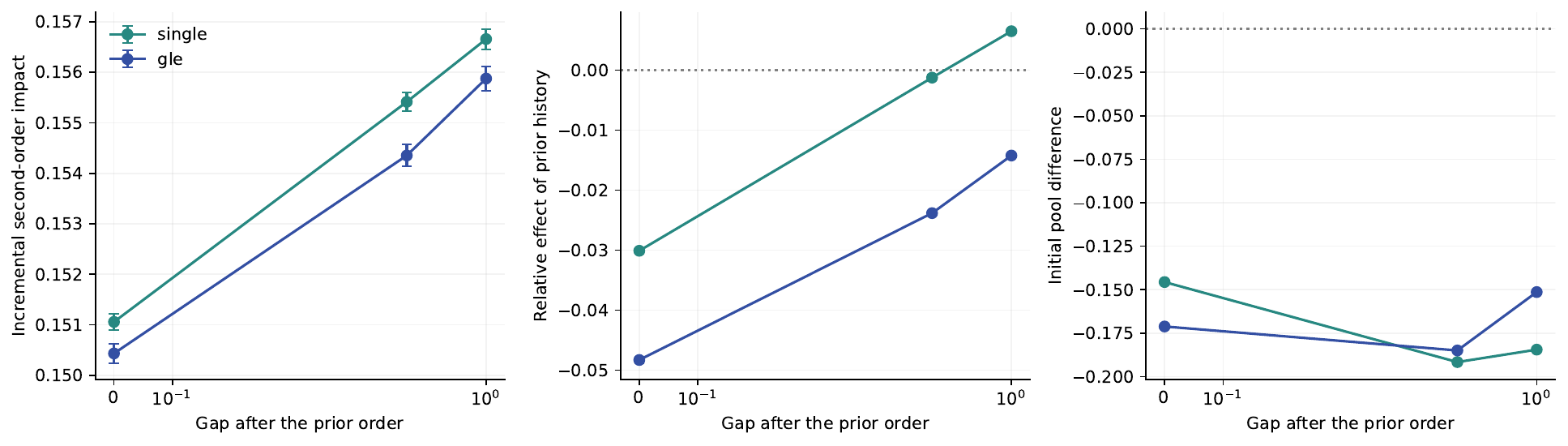}}{\fbox{\parbox[c][0.22\textheight][c]{0.92\textwidth}{\centering Figure placeholder: \texttt{figure\_3\_memory.pdf}\\[1em]Inherited-state response and memory-spectrum comparison}}}
\caption{Short-gap effect of a prior buy order on a subsequent buy order for the single-mode and GLE pools. Both orders have size $Q_0$ and duration $\tau _0$. From left to right, the panels show incremental probe impact $J_{2|1}(G)$, the total relative history effect $\mathcal H(G)$, and the initial opposing-pool difference $\Delta\rho_{\mathrm{init}}(G)$. The horizontal dotted lines mark zero history effect and zero pool difference. The residual displacement is retained, and each probe is measured against its matching no-probe continuation. The displayed gaps are $G / \tau _0 = 0$, $0.5$, and $1$. The middle panel shows the total effect, not the fresh-pool-adjusted quantity $\mathcal H_{\mathrm{pool}}$. The broad-spectrum comparison is reported separately in \cref{tab:memory-history}. Error bars are $\pm 1.96$ standard errors of the paired Monte Carlo estimates and are visible only in the left panel.}
\label{fig:memory}
\end{figure}

\Cref{fig:memory} shows a negative total history effect at zero gap for both pool models, $- 4.8\%$ for the GLE pool and $- 3.0\%$ for the single-mode pool. The effect weakens as the gap grows. At $G = \tau _0$ it is $- 1.4\%$ for the GLE pool and $+ 0.6\%$ for the single-mode pool, whose effect changes sign between $G = 0.5 \tau _0$ and $\tau _0$. The GLE value at $G = 0$ agrees with the decomposition above, $- 12.3\% + 7.5\%$. The incremental probe impact is slightly lower for the GLE pool than for the single-mode pool at every displayed gap, and their gap profiles differ. The initial opposing pool is depleted in both cases, by $0.15$ to $0.19$ in pool intensity. This pool difference is a diagnostic of the model's latent state. Empirical use would require an observable liquidity proxy. Its difference from the no-prior-order history first becomes more negative and then partly recovers, consistent with a delayed latent-state response after the prior order ends. Pool recovery and the incremental price response need not be monotone functions of the gap, because the probe inherits both displacement and latent memory.

The simultaneous negative history effect and negative pool difference illustrate why total probe impact cannot isolate depletion. The fresh-pool subtraction in \eqref{eq:pool-channel} is needed for the operational comparison of that contribution, but it is not displayed in the current figure. Nor do these short-gap panels show the long tail described by the separate broad-spectrum summary.

The broad-spectrum comparison uses four logarithmically spaced rates around each baseline geometric rate, with integrated strengths $S_{YY}^{\mathrm{int}}$ and $S_{YX}^{\mathrm{int}}$ preserved and zero-lag strengths left free. The half-decay times in \cref{tab:memory-history} are obtained by log-linear interpolation between the gap grid points. The fresh-pool-adjusted gap profiles subtract $\mathcal{H}_{\mathrm{fresh}}(G)$ from $\mathcal{H}(G)$ to give $\mathcal{H}_{\mathrm{pool}}(G)$ of \eqref{eq:pool-channel}.

\subsubsection{Depletion depth and spectral separation} \label{sec:inherited-depletion}

The additional comparison varies the prior buy size over $Q_{\mathrm{prior}}/Q_0 \in \{1,3,5\}$ while retaining a unit-size probe and constant-rate execution over $\tau _0$ for each order. The gap grid is $G/\tau _0 \in \{0,0.25,0.5,1,2,5,10,20,50\}$. The $81$ configurations use $2048$ paths, a time step of $0.01\tau _0$, and one Monte Carlo ensemble with a configuration-dependent seed per case. We write $\mathcal H(G;Q_{\mathrm{prior}})$ for the total relative history effect \eqref{eq:prior-history-effect} under the specified prior size. Its maximum $\mathcal H_{\max}$ over the nine sampled gaps and the maximizing gap $G_*$ are summarized in \cref{tab:inherited-depletion}.

\begin{table}[!htb]
\centering
\small
\begin{tabular}{llrrrrrr}
\hline
$Q_{\mathrm{prior}}/Q_0$ & Pool & $\mathcal H(0)$ & $\mathrm{SE}(\mathcal H(0))$ & $\mathcal H_{\max}$ & $\mathrm{SE}(\mathcal H_{\max})$ & $G_*/\tau _0$ & $\mathcal H(50\tau _0)$ \\
\hline
1 & Single & $- 3.00\%$ & $0.02\%$ & $+ 0.65\%$ & $0.02\%$ & $1$ & $- 0.26\%$ \\
 & GLE & $- 4.83\%$ & $0.03\%$ & $- 0.25\%$ & $<0.01\%$ & $50$ & $- 0.25\%$ \\
 & Broad & $- 6.10\%$ & $0.03\%$ & $- 0.17\%$ & $<0.01\%$ & $50$ & $- 0.17\%$ \\
\hline
3 & Single & $+ 5.35\%$ & $0.06\%$ & $+ 8.50\%$ & $0.05\%$ & $1$ & $- 0.26\%$ \\
 & GLE & $+ 0.80\%$ & $0.06\%$ & $+ 4.01\%$ & $0.03\%$ & $2$ & $- 0.25\%$ \\
 & Broad & $- 2.02\%$ & $0.06\%$ & $+ 1.60\%$ & $0.03\%$ & $2$ & $- 0.01\%$ \\
\hline
5 & Single & $+ 9.23\%$ & $0.07\%$ & $+ 13.24\%$ & $0.06\%$ & $1$ & $- 0.26\%$ \\
 & GLE & $+ 3.45\%$ & $0.07\%$ & $+ 7.22\%$ & $0.04\%$ & $2$ & $- 0.25\%$ \\
 & Broad & $- 0.21\%$ & $0.07\%$ & $+ 3.83\%$ & $0.04\%$ & $2$ & $+ 0.15\%$ \\
\hline
\end{tabular}
\caption{Total relative history effect as the prior-order size increases,
with paired delta-method standard errors. The effect includes both residual
displacement and the evolving latent pool. Maxima and their locations refer
to the sampled gap grid, not to continuous-gap optimization. The reported
$\mathrm{SE}(\mathcal H_{\max})$ is the standard error of $\mathcal H$
evaluated at the maximizing gap $G_*$. At prior size $Q_0$, the GLE and
broad-spectrum maxima occur at the upper grid boundary and are not interior
peaks.}
\label{tab:inherited-depletion}

\end{table}

The range of the three maxima grows from $0.90$ percentage points at $Q_{\mathrm{prior}} = Q_0$ to $6.90$ at $3Q_0$ and $9.41$ at $5Q_0$, with each extreme resolved at better than $0.06$ percentage points by its paired delta-method standard error. At the two larger sizes, the single-mode maximum is largest, the GLE maximum is intermediate, and the broad-spectrum maximum is smallest, as shown in \cref{fig:inherited-depletion}. The sign pattern also depends on the gap. The baseline single-mode effect is already positive near $G = \tau _0$, while the baseline GLE and broad-spectrum effects remain negative. For the broad spectrum, the effect at zero gap remains negative even at the larger prior sizes. Its first positive sampled value occurs at $G = 0.5\tau _0$ for $3Q_0$ and $G = 0.25\tau _0$ for $5Q_0$. The paired standard errors are $0.02$--$0.07$ percentage points across the grid, so the larger prior orders separate the spectra in amplitude. At $Q_{\mathrm{prior}} = Q_0$ and $G = 50\tau_0$, the paired standard errors of the GLE and broad-spectrum pools are below $0.01$ percentage points.

\begin{figure}[!htb]
\centering
\IfFileExists{figure_4_inherited_depletion.pdf}{\includegraphics[width=0.7\textwidth]{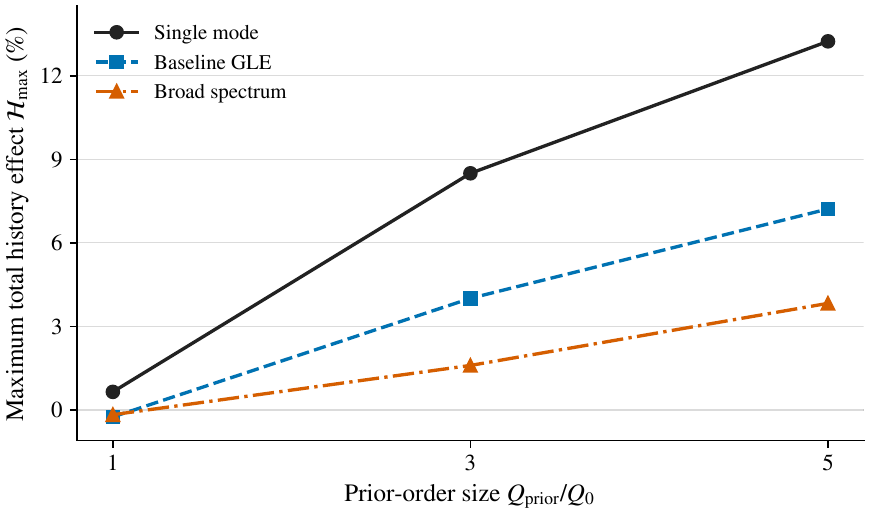}}{\fbox{\parbox[c][0.22\textheight][c]{0.92\textwidth}{\centering Figure placeholder: \texttt{figure\_4\_inherited\_depletion.pdf}}}}
\caption{Maximum total history effect over the reported gap grid versus prior-order size. Markers show the point estimates in \cref{tab:inherited-depletion}, connecting lines guide the eye. The paired standard errors in \cref{tab:inherited-depletion} are at most $0.06$ percentage points, smaller than the markers.}
\label{fig:inherited-depletion}
\end{figure}

The initial pool differences at zero gap range from $- 0.428$ to $- 0.370$ for $3Q_0$ and from $- 0.535$ to $- 0.493$ for $5Q_0$. With the symmetric no-prior mean equal to one, these correspond approximately to mean pool intensities of $0.57$--$0.63$ and $0.47$--$0.51$, respectively, above the floor $0.3$. The comparison therefore tests substantial but incomplete depletion, not suppression of spectral differences by the floor. Indeed, the broad pool is more depleted at zero gap than the single-mode pool but has a smaller total history effect. This ordering depends on the inherited displacement and subsequent pool evolution, not on initial depletion alone. The unit-prior-size rows at gaps $0$, $0.5\tau_0$, and $\tau_0$ reproduce the short-gap values behind \cref{fig:memory}, confirming that the two comparisons are computed under the same state dynamics.

At both larger prior sizes, the sampled maxima occur at $G = \tau _0$ for the single-mode pool and $G = 2\tau _0$ for the other spectra. Their locations remain sensitive to sampling and gap resolution: at $5Q_0$, the GLE values at gaps $\tau _0$ and $2\tau _0$ are $7.21\%$ and $7.22\%$, and the broad-spectrum values are $3.81\%$ and $3.83\%$. The paired standard errors at these near-tied gaps are $0.05\%$ and $0.04\%$, comparable to the amplitude difference, so the peak location is not resolved at the reported path count. The larger differences in amplitude show greater sensitivity to the specified spectra under stronger prior forcing, and the paired standard errors in \cref{tab:inherited-depletion} confirm that these amplitude differences are resolved above the paired sampling noise at the reported path count.

\subsection{Interpretation and scope of the results} \label{sec:robustness}

The comparisons show how depletion changes the size--duration response and how execution timing and prior trading affect impact. They use a symmetric potential, constant cross-memory amplitudes, fixed ambient volatility, and finite-start initialization. Robustness to alternative structural parameters and initialization remains to be assessed.

The spectral comparison preserves integrated kernel strengths while allowing zero-lag strengths and force variance to change. The observed differences therefore reflect the chosen kernel specifications. Whether these differences are detectable in market data, or permit unique identification of the kernel parameters, requires empirical investigation.

The comparisons retain the balanced reference \eqref{eq:reference}. Adding the same drift $\mu \, dt$ to price and reference leaves the displacement equation unchanged when its other coefficients depend only on displacement, latent state, and the prescribed order. Adding drift only to the reference changes the relative dynamics and defines a different model.

\section{Discussion and conclusions} \label{sec:conclusions}

The model hierarchy shows that an extended concave impact regime, including an approximately square-root interval, can emerge without prescribing a square-root dependence of impact on order size. Under the stated assumptions on individual trading responses and activation thresholds, the fresh-pool benchmark yields an intermediate square-root regime between linear small- and large-order limits. The numerical comparisons show that this regime persists under dynamic depletion and memory, although its extent and magnitude change. These results distinguish the counterflow mechanism generating the regime from the liquidity dynamics that modify it. The exponent and its range remain conditional on the behavioral assumptions. Normalizing thresholds and individual responses by the noise scale gives volatility proportionality within the square-root regime, while thresholds tied to the execution horizon remove duration dependence from the fresh-pool square-root law.

Memory acts through the pool size, affecting the square-root prefactor, the extent of depletion, and post-execution recovery. A positive floor prevents complete pool depletion, without supplying a general concavity theorem. The numerical comparisons show a narrower square-root interval under depletion and often similar single-order responses for the matched single-mode and baseline GLE pools. This supports the simpler specification as a useful approximation for those outputs at the reported settings, while the quasi-static prefactor remains spectrum-dependent through the law of the latent state. Schedule timing changes both terminal impact and execution-weighted displacement, with different rankings. A prior order leaves residual displacement as well as a depleted pool. Their effects on a subsequent probe can oppose one another. Increasing the prior order to $3Q_0$ and $5Q_0$ enlarges the spread of the sampled maximum total history effects to $6.90$ and $9.41$ percentage points, respectively. The baseline broad spectrum also has a larger reported long-gap pool-associated difference. The paired standard errors resolve these amplitude differences, while the peak locations remain unresolved on the sampled gap grid. These comparisons motivate conditional and repeated-order observables as additional sources of information about memory. They do not show that single-order impact is generally independent of the spectrum, nor that multi-order observations identify it uniquely.

Constant displayed depth and counterflow opposing displacement imply nonnegative expected round-trip costs for the log-price functional \eqref{eq:cost}, independently of the memory spectrum. This is the no-manipulation result under the stated first-order cost convention. If displayed depth were replaced by $\phi(Y_t)$, the identity would acquire the term $-\EE\int_0^H \tfrac12 D_t^2\,d\phi(Y_t)$, which can be negative. Extending the result to state-dependent depth would therefore require additional restrictions, see \cite{FruthEtAl2014,AlfonsiAcevedo2014}.

Three limitations bear on the interpretation. First, the balanced reference leaves no permanent impact. Under quadratic onset with fixed thresholds, a buy order that has ended satisfies $1/D_{T+\tau} = 1/D_T + (\omega/L_0)\int_T^{T+\tau}\rho(Y_s)\,ds$ on every path, so the floor $\rho_{\mathrm{fl}} > 0$ drives the displacement to zero. For the fresh pool the late residual is $L_0/(\omega\tau)$, independent of order size. For the stochastic pool this limit also requires the time average of $\rho(Y_s)$ to approach one. Fair-pricing arguments instead predict a permanent component proportional to peak impact \cite{FarmerEtAl2013}. Such a component would require the reference $m_t$ to respond to the order. Second, the fixed detection horizon of the numerical study makes terminal impact depend on the execution rate and gives a nearly flat path after a short transient, whereas elapsed-time thresholds give square-root growth in executed volume. A single specification with both properties is left open. Third, the latent state is forced by submitted flow in reduced form. Forcing it by realized counterflow would tie depletion to absorbed volume, at the cost of feedback in the forcing.

The reported stochastic comparisons use Monte Carlo, alongside analytical and deterministic benchmarks. \Cref{app:pinn} records a backward PINN formulation verified against Monte Carlo at the baseline. Further numerical development, empirical calibration, and validation are left to a companion study. The empirical task includes distinguishing induced counterflow from background trading, specifying observable liquidity proxies, and assessing the counterfactual reference. The present contribution is the model, its conditional analytical results, and numerical comparisons illustrating how latent-pool dynamics modify impact.

\section*{Disclosure statement}

No potential conflict of interest was reported by the author.

\section*{Funding}

No funding was received.

\section*{Disclaimer}

The opinions expressed here are the author's own and do not represent the views of the author's employer.

\section*{Acknowledgments}

I thank Michael Isichenko for various discussions on the counterflow mechanism.

\printbibliography[title={References}]

\appendix
\label{app:proofs}
\appendixpage
\numberwithin{equation}{section}
\setcounter{equation}{0}

\section{Exactness of the Markovian lift (\cref{prop:lift})}  \label{sec:lift-appendix}

\begin{proof}

From \eqref{eq:etastate},
\begin{equation}
d \eta _{i,t} = - \gamma _i \eta _{i,t}\,dt + \sigma _{Y,i}\,dW_{i,t}.
\end{equation}
Differentiating \eqref{eq:hstate} therefore gives the first equation in
\eqref{eq:hg-ode}. Differentiating \eqref{eq:gstate} gives the second. Moreover,
\begin{equation}
\sum_{i = 1}^N a_i h_{i,t} = \int_0^t K_{YY}(t - s)\, dY_s - \sum_{i = 1}^N a_i \eta _{i,t}, \qquad \sum_{j = 1}^M g_{j,t} = \int_0^t K_{YX}(t - s;Y_s)q_s\,ds.
\end{equation}
Substitution into \eqref{eq:gle-lift} recovers \eqref{eq:gle} with $\xi _t = \sum_i
a_i \eta _{i,t}$.

It remains to verify the FDT. The stationary OU initialization in
\eqref{eq:etastate} gives, for all $s,t\ge0$,
\begin{equation}
\EE[ \eta _{i,t} \eta _{i,s}] = \frac{ \sigma _{Y,i}^2}{2 \gamma _i}e^{ - \gamma _i|t - s|}.
\end{equation}

Independence of the Brownian motions and the choice $\sigma _{Y,i}^2 = 2 \gamma _i
 \sigma _Y^2/a_i$ then yield
\begin{equation}
\EE[ \xi _t \xi _s] = \sum_{i = 1}^N a_i^2 \frac{ \sigma _{Y,i}^2}{2 \gamma _i}e^{ - \gamma _i|t - s|} = \sigma _Y^2 \sum_{i = 1}^N a_i e^{ - \gamma _i|t - s|} = \sigma _Y^2 K_{YY}(|t - s|),
\end{equation}
which is \eqref{eq:fdt}. Hence the finite-dimensional system reproduces both the
friction term and the stationary colored force exactly.

\end{proof}

\section{The counterflow construction (\cref{prop:cf-structure})} \label{sec:cf-proofs}

\begin{proof}

Oddness follows from the factor $\operatorname{sgn}(D)$ in \eqref{eq:cf-aggregate}. For $D\ge0$ write $\mathcal A(D) = \int_{[0,\infty)}(D - \chi )_ + \, \Pi (d \chi )$, where $x_ + = \max(x,0)$. For each $\chi$ the integrand is nondecreasing and convex in $D$, with right derivative $\ind_{\{ \chi \le D\}}$. Dominated convergence, justified by the finite mass of $\Pi$, gives $\mathcal A'(D) = \Pi ([0,D])\ge0$. This right derivative is nondecreasing, so $\mathcal A$ is convex on $[0,\infty)$. This proves (i).

If $\pi _0 = 0$, then for $D > 0$
\begin{equation*}
\mathcal A(D) = \int_0^D(D - \chi )\, \pi ( \chi )\,d \chi = \pi (0)\frac{D^2}{2} + \int_0^D(D - \chi )\big( \pi ( \chi ) - \pi (0)\big)\,d \chi .
\end{equation*}
The last integral is bounded in magnitude by $\tfrac12D^2\sup_{0\le \chi \le D}| \pi ( \chi ) - \pi (0)|$, which is $o(D^2)$ by continuity of $\pi$ at zero. Oddness extends the result to $D < 0$. This proves (ii).

For $D > 0$,
\begin{equation*}
\mathcal A(D) = | \Pi |(D - \bar \chi ) - \int_{(D,\infty)}(D - \chi )\, \Pi (d \chi ).
\end{equation*}
The last integral is bounded in magnitude by $\int_{(D,\infty)} \chi \, \Pi (d \chi )$, which tends to zero as $D\to\infty$ because the first moment is finite. Oddness again extends the result to $D < 0$. This proves (iii).

\end{proof}

\section{Well-posedness and round-trip cost (\cref{lem:wellposed,prop:roundtrip})} \label{sec:roundtrip-proof}

\begin{proof}[Proof of \cref{lem:wellposed}]
The noise enters only the $h$-equations, and it enters additively. The drift is locally Lipschitz in the state, uniformly on $[0,H]$. For the counterflow $\varphi(D,y) = \rho(\operatorname{sgn}(D)\,y)\,\mathcal A_t(D)$ this includes the crossing of $D = 0$. Indeed $\mathcal A_t(0) = 0$ and $|\mathcal A_t(D)|\le|\hat\Pi|\,|D|/s_t$. For $D$ and $D'$ of equal sign the bound follows from the Lipschitz continuity of $\rho$ and $\mathcal A_t$. For $D > 0 > D'$ it follows from $|\varphi(D,y) - \varphi(D',y')|\le|\varphi(D,y)| + |\varphi(D',y')|\le\sup\rho\,|\hat\Pi|\,|D - D'|/s_t$. A unique strong solution therefore exists up to an explosion time $\zeta$.

On $[0,\zeta)$, $\operatorname{sgn}(D_t)\,\varphi(D_t,Y_t)\ge0$. Hence $|D_t|$ is absolutely continuous with derivative at most $|q_t|/L_0$, which gives the bound on $D$. The states $g_{j,t}$ stay bounded by constants $\bar g_j$, because $c_j$ and $q$ are bounded.

Let $r_{i,t} = h_{i,t} + \eta_{i,t}$, with $\eta_{i,t}$ the stationary mode of \eqref{eq:etastate}. By \eqref{eq:hg-ode} the Brownian terms cancel and $\dot r_{i,t} = -\gamma_i r_{i,t} + \dot Y_t$. Write $\Phi_t = U'(Y_t) + \sum_i a_i r_{i,t}$ and $f_t = \xi_t + \sum_j g_{j,t}$, with $\xi_t = \sum_i a_i\eta_{i,t}$. Then \eqref{eq:gle-lift} reads $\dot Y_t = -\Phi_t + f_t$, and $V_t = U(Y_t) + \tfrac12\sum_i a_i r_{i,t}^2$ satisfies
\begin{equation*}
\dot V_t = \Phi_t\dot Y_t - \sum_i a_i\gamma_i r_{i,t}^2 = -\Phi_t^2 + \Phi_t f_t - \sum_i a_i\gamma_i r_{i,t}^2 \le \tfrac14 f_t^2.
\end{equation*}
Hence $V_t\le V_0 + \tfrac14\int_0^t f_s^2\,ds$, which is finite on every path because $f$ is continuous. Since $u_4 > 0$, there is a constant $C$ with $U(y)\ge u_4y^4/8 - C$. The processes $Y$ and $r$, and therefore $h$, stay bounded on $[0,\zeta\wedge H)$, so $\zeta > H$.

Under the initial law of \cref{prop:lift}, $Y_0 = 0$ and $r_{i,0} = h_{i,0} + \eta_{i,0} = 0$, so $V_0 = U(0) = 0$. With $\EE[\xi_t^2] = \sigma_Y^2K_{YY}(0)$ from \eqref{eq:fdt},
\begin{equation*}
\EE[V_t]\le\tfrac12\int_0^t\Big(\sigma_Y^2K_{YY}(0) + \big(\textstyle\sum_j\bar g_j\big)^2\Big)\,ds,
\end{equation*}
and the fourth-moment bound follows from $U(y)\ge u_4y^4/8 - C$.
\end{proof}

\begin{proof}[Proof of \cref{prop:roundtrip}]

Let $x_t = \int_0^t q_s\,ds$ denote the position, so that $x_0 = x_H = 0$ and $x$ is bounded and predictable. Since $D_0 = 0$, the price and the reference start together, $X_0 = m_0$, and $X_t - X_0 = D_t + (m_t - m_0)$ for all $t$. By \eqref{eq:reference} and integration by parts,
\begin{equation*}
\int_0^H q_t\,(m_t - m_0)\,dt = \big[x_t\,(m_t - m_0)\big]_0^H - \int_0^H x_t\,\sigma _t\,dW_t^X = - \int_0^H x_t\,\sigma _t\,dW_t^X .
\end{equation*}
The last integral has zero expectation because $x$ is bounded and the price-noise integral is a square-integrable martingale. The displacement is absolutely continuous, and \eqref{eq:displacement-ode} gives $q_t = L_0\dot D_t + q_t^{\mathrm{cf}}$. Hence
\begin{equation*}
\int_0^H q_t\,D_t\,dt = L_0\int_0^H D_t\,\dot D_t\,dt + \int_0^H q_t^{\mathrm{cf}}D_t\,dt = \frac{L_0}{2}D_H^2 + \int_0^H q_t^{\mathrm{cf}}D_t\,dt .
\end{equation*}
Taking expectations gives \eqref{eq:roundtrip-identity}. Both terms are nonnegative because $q_t^{\mathrm{cf}}D_t\ge0$. Under \eqref{eq:qcf} this sign condition holds for every $\rho > 0$, because $\mathcal A_t$ has the sign of its argument.

\end{proof}

\section{Proof of \cref{prop:sqrt-cf}} \label{sec:proof4}

\begin{proof}

Let $q = Q/T$ and $D_\infty = \sqrt{q/ \omega }$. Since $D_0 = 0 < D_\infty$ and the right side of \eqref{eq:cf-ode} is positive on $[0,D_\infty)$, the solution remains in this interval and satisfies $L_0\dot D_t = \omega (D_\infty^2 - D_t^2)$. Separation of variables gives
\begin{equation*}
D_t = D_\infty\tanh\!\left(\frac{ \omega D_\infty t}{L_0}\right).
\end{equation*}
Setting $t = T$ and using $\omega D_\infty T = \sqrt{ \omega QT}$ yields \eqref{eq:sqrt-cf}. As $Q\to0$, $\tanh x = x + O(x^3)$ gives $\mathcal I_T(Q) = Q/L_0 + O(Q^2)$, so $\delta _{\mathrm{eff}}\to1$. As $Q\to\infty$, the hyperbolic tangent tends to one exponentially fast in $\sqrt Q$, so $\log\mathcal I_T(Q) = \tfrac12\log Q + \mathrm{const} + o(1)$ and its derivative with respect to $\log Q$ tends to one half.

\end{proof}

\section{Duration-free thresholds (\cref{rem:duration-free})} \label{sec:duration-free-proof}

With a fresh pool, quadratic onset, and a constant rate $q = Q/T$, the displacement solves $L_0\dot D_t = q - \hat\pi(0)\,D_t^2/(2s_t^2)$ with $D_0 = 0$.

For $\tau _d = c_\tau T$ the threshold scale $s = \sigma\sqrt{c_\tau T}$ is constant, and \cref{prop:sqrt-cf} applies with $\omega = \hat\pi(0)/(2\sigma^2c_\tau T)$. Then $\omega T = \hat\pi(0)/(2\sigma^2 c_\tau)$, so that $\sqrt{Q/(\omega T)} = \sigma\sqrt{2c_\tau Q/\hat\pi(0)}$ and $\sqrt{\omega QT}/L_0 = x$. This gives the first expression in \eqref{eq:duration-free}.

For $s_t = \sigma\sqrt{c_\tau t}$, write $\alpha_q = q/L_0$ and $\beta_q = \hat\pi(0)/(2\sigma^2c_\tau L_0)$, so that $\dot D_t = \alpha_q - \beta_q D_t^2/t$. The substitution $D_t = t\,\dot w_t/(\beta_q w_t)$ turns this Riccati equation into the linear equation $t\,\ddot w_t + \dot w_t - \alpha_q\beta_q w_t = 0$. Its solution regular at zero is $w_t = I_0\big(2\sqrt{\alpha_q\beta_q t}\big)$, which gives
\begin{equation*}
D_t = \sqrt{\frac{\alpha_q t}{\beta_q}}\; \frac{I_1\big(2\sqrt{\alpha_q\beta_q t}\big)}{I_0\big(2\sqrt{\alpha_q\beta_q t}\big)} .
\end{equation*}
This solution vanishes at zero with $\dot D_0 = \alpha_q$. At $t = T$, $\alpha_q T/\beta_q = 2\sigma^2c_\tau Q/\hat\pi(0)$ and $\alpha_q\beta_q T = x^2$, which gives the second expression in \eqref{eq:duration-free}. Both expressions depend on $Q$ but not on $T$. The expansion $I_1(z)/I_0(z) = 1 - 1/(2z) + O(z^{-2})$ as $z\to\infty$ gives the algebraic approach stated in \cref{rem:duration-free}.

\section{Pathwise bounds (\cref{prop:pool-bounds})} \label{sec:pool-bounds-proof}

\begin{proof}

For $r\ge0$ and $D\ge0$ define $f(t,D;r) = \big[q - r\,\mathcal A_t(D)\big]/L_0$. Since $\mathcal A_t(0) = 0$, the drift of the displacement at $D = 0$ equals $q/L_0\ge0$, so $D_t\ge0$ on $[0,T]$, and the same holds for each $D^{(r)}$. For $D\ge0$ the function $f$ is nonincreasing in $r$, because $\mathcal A_t(D)\ge0$. Along the fixed path, $D$ solves $\dot D_t = f\big(t,D_t; \rho(Y_t)\big)$, and
\begin{equation*}
f(t,D;\rho_{\max})\le f\big(t,D;\rho(Y_t)\big)\le f(t,D;\rho_{\min}), \qquad D\ge0.
\end{equation*}
Since $f$ is locally Lipschitz in $D$, the comparison theorem for scalar ordinary differential equations gives \eqref{eq:pool-bounds}. The closed form follows from \cref{prop:sqrt-cf} with $\omega$ replaced by $r\omega$. The same argument on $[T,T+\tau]$ with $q = 0$, started from the realized $D_T$, gives \eqref{eq:pool-relaxation-bounds} through the relaxation formula \eqref{eq:quadratic-relaxation}.

\end{proof}

\section{Backward equations and a PINN approximation} \label{app:pinn}

The stochastic results in the main text use Monte Carlo simulation of the lifted system. This appendix records a backward formulation of the same response functionals, its approximation by physics-informed neural networks (PINN), and a verification against Monte Carlo at the baseline. It is intended as a starting point for the companion study.

\subsection{Backward equations for model response functionals} \label{subsec:backward-observables}

The generator of the augmented process of \cref{subsec:numerical-lift} acts on a sufficiently smooth function $f$ according to
\begin{align} \label{eq:numerical-generator}
\mathcal L_t^q f &= b_q \partial_D f + F \partial_y f + \sum_{i = 1}^N (F - \gamma _i h_i) \partial_{h_i} f + \sum_{j = 1}^M \left[ - \lambda _j g_j + c_j(y) q_t\right] \partial_{g_j} f + \frac{1}{2} \sum_{i = 1}^N \sigma _{Y,i}^2 \partial_{h_i h_i} f.
\end{align}
There are no mixed diffusion terms, either between the displacement and the memory states or between distinct memory states. The generator is degenerate because $D$, $Y$, and the order-flow memory states have no direct Brownian increments. A direct algebraic check of the generator follows from applying it to $f(\mathbf z) = D h_i$,
\begin{equation} \label{eq:numerical-generator-cross-check}
\mathcal L_t^q(D h_i) = b_q h_i + D(F - \gamma _i h_i),
\end{equation}
which contains no second-order contribution.

When $c_j$ is constant and $g_{j,0}$ is deterministic, $g_j(t)$ in \eqref{eq:numerical-deterministic-flow-memory} is substituted directly into the coefficients of the backward equation and need not be included as a PINN state variable. State-dependent amplitudes $c_j(Y_t)$ generally prevent this reduction.

The numerical experiments require both expectations of state variables at specified times and expectations of quantities accumulated over an interval. These objects can be computed from a common backward equation.

For an observation horizon $H$, terminal function $G$, and running function $r$, define
\begin{equation} \label{eq:general-backward-observable}
u_{G,r}^q(t,\mathbf z;H) = \EE_{t,\mathbf z}^q \left[G(\widetilde{\mathbf Z}_H) + \int_t^H r(s,\widetilde{\mathbf Z}_s) \, ds\right].
\end{equation}
The function $u_{G,r}^q$ satisfies
\begin{equation} \label{eq:backward-observable-pde}
\partial_t u_{G,r}^q + \mathcal L_t^q u_{G,r}^q + r = 0, \qquad u_{G,r}^q(H,\mathbf z;H) = G(\mathbf z).
\end{equation}
No discounting is introduced because these equations compute physical expectations of model responses rather than prices of traded claims.

Let $\nu _0$ denote the initial distribution of the latent state and memory variables. Define the initial-state averaging operator by
\begin{equation} \label{eq:initial-average-operator}
\mathcal M_{ \nu _0}[u] = \int u(0,\mathbf z) \, \nu _0(d\mathbf z).
\end{equation}
Every order experiment is paired with a no-order calculation using the same model parameters, initial distribution, computational domain, and numerical settings. For a terminal function $G$, the order-induced response at horizon $H$ is therefore
\begin{equation} \label{eq:paired-terminal-response}
\Delta _G^q(H) = \mathcal M_{ \nu _0}\left[u_{G,0}^q(\cdot;H) - u_{G,0}^0(\cdot;H)\right].
\end{equation}
This subtraction removes changes generated by the uncontrolled stochastic dynamics and isolates the response associated with the prescribed order.

The principal price response is the expected order-induced displacement,
\begin{equation} \label{eq:output-impact}
J_q(t) = \Delta _D^q(t).
\end{equation}
The expected cumulative signed counterflow volume is obtained by setting
\begin{equation} \label{eq:counterflow-running-volume}
r_{\mathrm{cf}}(t,\mathbf z) = \rho\big(\operatorname{sgn}(D)\,y\big)\,\mathcal A_t(D)
\end{equation}
and defining
\begin{equation} \label{eq:output-counterflow-volume}
V_q^{\mathrm{cf}}(H) = \mathcal M_{ \nu _0}\left[u_{0,r_{\mathrm{cf}}}^q(\cdot;H) - u_{0,r_{\mathrm{cf}}}^0(\cdot;H)\right].
\end{equation}
Because the depth is constant, the accumulated contribution of the counterflow to log-price displacement is $V_q^{\mathrm{cf}}(H)/L_0$.

Second moments can be obtained by choosing terminal functions such as $G(\mathbf z) = D^2$ or $G(\mathbf z) = y^2$. Consequently, the backward solver provides a common computational framework for all expectations and accumulated quantities used in the main text. Arbitrary distributional quantiles are not inferred from these moment calculations. If such quantities are required, the corresponding distributional problem must be solved separately.

For post-execution calculations, the prescribed order is set to zero after the execution horizon $T$, and the backward problem is solved for horizons $H > T$. Repeating this calculation over a sequence of horizons generates the post-execution response curves.

\subsection{PINN approximation} \label{subsec:pinn-solver}

We approximate the backward equations by physics-informed neural networks (PINN), following the numerical approach of \cite{ItkinKazbekGLE}. The equations here describe physical expectations under prescribed order flow. No market-data calibration objective is included in the present study.

The full state dimension is $N + M + 2$, reduced to $N + 2$ when the order-flow memory states are deterministic and eliminated by \eqref{eq:numerical-deterministic-flow-memory}. A PINN represents the conditional response without a tensor-product grid. This can be useful for repeated evaluation over initial states, although the baseline comparisons below show no computational advantage over Monte Carlo. A network taking model parameters as additional inputs would require separate training and validation and is not used for the reported sweeps.

After eliminating deterministic order-flow memory variables whenever \eqref{eq:numerical-deterministic-flow-memory} applies, the terminal condition is imposed exactly through
\begin{equation} \label{eq:pinn-terminal-ansatz}
u_ \theta (t,\mathbf z;H) = G(\mathbf z) + (H - t) N_ \theta (t,\mathbf z;H),
\end{equation}
where $N_ \theta$ is a neural network with smooth activation functions. Time, state variables, and outputs are scaled before training. Automatic differentiation evaluates all derivatives entering \eqref{eq:numerical-generator}, including the second derivatives in the memory states.

The PDE residual is
\begin{equation} \label{eq:pinn-residual}
\mathcal R_ \theta (t,\mathbf z) = \partial_t u_ \theta (t,\mathbf z;H) + \mathcal L_t^q u_ \theta (t,\mathbf z;H) + r(t,\mathbf z),
\end{equation}
and the interior training loss is
\begin{equation} \label{eq:pinn-loss}
\mathcal J_{\mathrm{PDE}}( \theta ) = \frac{1}{N_c} \sum_{k = 1}^{N_c} w_k \left|\mathcal R_ \theta (t_k,\mathbf z_k)\right|^2, \qquad w_k > 0.
\end{equation}
Independent residual evaluations can guide collocation refinement, especially in regions contributing to initial-state averages or reached by large orders. The baseline comparisons below and in \cref{app:pinn-convergence} supply the reported verification. Systematic residual-based refinement over a parameter domain is not claimed here.

For schedules with discontinuities, the time domain is divided at the switching times, including the end of execution. The solution is matched continuously between adjacent intervals. Its time derivative need not be continuous when the drift changes discontinuously.

The reported impact exponents are obtained from separately computed order sizes. They are not derivatives of a trained parametric network.

\subsection{State-space coverage and initial-state averaging} \label{subsec:pinn-domain}

The backward problem is posed on an unbounded state space, whereas the baseline PINNs are trained on finite regions constructed from pilot simulations. Their reported comparisons concern the sampled initial states at the tested schedules and horizons, not arbitrary states or parameter values.

A finite collocation region does not by itself define a boundary condition. Any imposed spatial boundary condition must be justified by the response functional. The baseline formulation imposes the terminal condition exactly and does not introduce arbitrary zero-value or reflecting spatial conditions. Finite-domain coverage remains an approximation limitation.

For the finite-start law, initial-state averaging is an integral over Gaussian modal variables. Gauss--Hermite or sparse-grid quadrature provides deterministic approximations to this integral, with a quadrature error distinct from network approximation error. Gaussian initial variables do not make these averages exact for a nonlinear network. The reported PINN verification instead evaluates the network on the same $2048$ initial states as the independent Monte Carlo reference, so that the comparison does not confound their different initial samples.

An equilibrium latent ensemble would require the joint invariant law of $Y$ and its memory states. It is not used in the reported study, which retains the finite-start law throughout.

\subsection{Baseline verification} \label{subsec:pinn-baseline}

The flat baseline case is checked against independent Monte Carlo simulation of the lifted system, with further schedule--horizon comparisons in \cref{app:pinn-convergence}. The comparison is paired. The network is evaluated at the initial states of the reference ensemble, and the uncertainty is the standard error of the per-state differences, which excludes the initial-state variability common to both estimators. Let $\Delta D_T$ and $\Delta V_T$ denote the differences between an order and its control in the terminal displacement and in the accumulated counterflow volume $\int_0^T q_t^{\mathrm{cf}}\,dt$. By \eqref{eq:displacement-ode}, $\Delta D_T + \Delta V_T/L_0 = Q/L_0$ holds path by path, and it holds exactly in the discretized scheme as well. The residual of this identity in the PINN outputs measures their mutual consistency without additional Monte Carlo sampling. It does not bound the individual output errors: displacement and counterflow errors can cancel in the identity. The independent reference comparison supplies separate accuracy evidence. The reported paired standard errors quantify sampling uncertainty. Time-discretization bias and finite state-domain coverage remain separate limitations of this comparison. At the baseline capacity the counterflow is stiff, and the reference uses a time step of $0.01\,\tau _0$. The PINN domain is built from a pilot of comparable size, so that the reference initial states lie inside it.

The comparison at the baseline of \cref{tab:baseline-parameters} is summarized in \cref{tab:pinn-verification}. The identity residual is $1.19\times10^{-3}$ in absolute units, a single number because both observables enter the same algebraic identity. The no-order control has zero displacement and zero accumulated counterflow on every simulated path. The trained control network reproduces this with maximum absolute errors of $4.0\times10^{-3}$ and $3.9\times10^{-3}$ over the $2048$ reference states, below the training tolerance of $5\times10^{-3}$, so the order-minus-control difference isolates the order's response to within the approximation error of the network.

The paired errors are $-4.6\times10^{-4}$ for $D$ and $+1.0\times10^{-3}$ for $V^{\mathrm{cf}}$, each resolved above the paired sampling error of $8.5\times10^{-5}$ and $8.4\times10^{-5}$ respectively. In relative terms the PINN reproduces the Monte Carlo reference to within $0.29\%$ on $D$ and $0.13\%$ on $V^{\mathrm{cf}}$. Training cost is $132.6$\,s for the order network and $141.4$\,s for the control network, with the reference and pilot Monte Carlo runs a small fraction of that. The comparison records baseline accuracy and computational cost. The same check at additional schedules and horizons is reported in \cref{app:pinn-convergence}.

\begin{table}[!htb]
\centering
\small
\begin{tabular}{lcccc}
\hline
Observable & PINN & MC & Paired error & Relative \\
\hline
$D$              & $0.157808$ & $0.158264$ & $-4.56\times10^{-4}$ $(8.48\times10^{-5})$ & $-0.288\%$ \\
$V^{\mathrm{cf}}$ & $0.842775$ & $0.841736$ & $+1.04\times10^{-3}$ $(8.36\times10^{-5})$ & $+0.124\%$ \\
\hline
\multicolumn{5}{l}{\footnotesize Identity residual (max): $1.19\times10^{-3}$.} \\
\hline
\end{tabular}
\caption{PINN verification at the baseline of \cref{tab:baseline-parameters}. The identity residual is the maximum of $|\Delta D_T + \Delta V_T/L_0 - Q/L_0|$ over the $2048$-state reference ensemble. It is a single number because both observables enter the same algebraic identity. Paired errors are PINN minus the independent Monte Carlo reference, with the standard error of the per-state differences in parentheses. Costs are wall-clock seconds on the same machine for the two-network run at $2000$ Adam epochs and $100$ L-BFGS steps per network.}
\label{tab:pinn-verification}

\end{table}

The network used for the reported run is a fully connected feedforward architecture with $\tanh$ activations, width $64$, and depth $3$, trained with Adam at learning rate $10^{-3}$ for $2000$ epochs per time block, followed by up to $100$ L-BFGS steps when polishing is enabled. Inputs are normalized so that time lies in $[-1,1]$ and each state component lies in $[-1,1]$ within the collocation box. The collocation domain is built from a pilot Monte Carlo cloud expanded by a factor of $1.5$ about its center; half of every batch is drawn from the cloud and half uniformly from the box. For general initial-state averaging, the solver provides Gauss--Hermite quadrature over the Gaussian modes of the finite-start law. The values in \cref{tab:pinn-verification} use the paired reference ensemble.

\subsection{Verification across schedules and horizons} \label{app:pinn-convergence}

The comparison in \cref{subsec:pinn-baseline} concerns the flat schedule at the unit horizon. The additional runs below test two schedule shapes and two horizons and include an analytical Kyle check. The Monte Carlo references use $2048$ paths and a time step of $0.01\tau _0$, each network uses $2000$ Adam epochs followed by $100$ L-BFGS steps. These are checks across configurations at fixed numerical settings, rather than a convergence study under refinement.

\begin{table}[!htb]
\centering
\small
\begin{tabular}{lcccc}
\hline
Schedule & Horizon & $D$ paired error & $V^{\mathrm{cf}}$ paired error & identity residual \\
\hline
Flat  & $1.0\,\tau_0$ & $-4.56\times10^{-4}$ $(8.5\times10^{-5})$ & $+1.04\times10^{-3}$ $(8.4\times10^{-5})$ & $1.19\times10^{-3}$ \\
Flat  & $0.5\,\tau_0$ & $-8.64\times10^{-4}$ $(3.4\times10^{-5})$ & $+9.52\times10^{-4}$ $(3.5\times10^{-5})$ & $9.9\times10^{-4}$ \\
Pause & $1.0\,\tau_0$ & $-2.05\times10^{-3}$ $(9.6\times10^{-5})$ & $+1.76\times10^{-3}$ $(9.6\times10^{-5})$ & $2.9\times10^{-3}$ \\
Pause & $0.5\,\tau_0$ & $-9.24\times10^{-4}$ $(2.7\times10^{-5})$ & $+1.41\times10^{-3}$ $(2.7\times10^{-5})$ & $1.4\times10^{-3}$ \\
\hline
Kyle  & $1.0\,\tau_0$ & \multicolumn{2}{c}{error against analytic $= 1.1\times10^{-4}$} & --- \\
\hline
\end{tabular}
\caption{PINN verification across schedules and horizons at the baseline of
\cref{tab:baseline-parameters}. The paired errors are PINN minus the
independent Monte Carlo reference, with the standard error of the per-state
differences in parentheses. The identity residual is the maximum of
$|\Delta D_T + \Delta V_T/L_0 - Q/L_0|$ over the reference ensemble and is
a single number because both observables enter the same algebraic identity.
The Kyle row is analytic.}
\label{tab:pinn-convergence}

\end{table}

Across the four stochastic settings, the absolute paired errors range from approximately $4.6\times10^{ - 4}$ to $2.1\times10^{ - 3}$ in the reported units. The largest occurs for the interrupted schedule at the longer horizon. The comparison therefore extends the evidence beyond the flat unit-horizon case, while remaining specific to these configurations and the stated numerical settings.

The identity residual is below $3\times10^{ - 3}$ throughout the grid and is largest for that same interrupted case. Schedule switches require matching the backward solutions across time blocks, but the table does not isolate interface error from the other approximation errors. The identity residual measures consistency between displacement and accumulated counterflow. Cancellation between their errors prevents its use as a bound on either output. Accuracy relative to the independent reference is assessed through the paired errors and their sampling uncertainties. Time-discretization bias and finite-domain effects remain separate.

The Kyle check gives an absolute error of $1.1\times10^{ - 4}$ against the analytical response. Because the displacement in this case is deterministic, the reference has no Monte Carlo sampling error. The check tests the solver against a closed form, but does not remove network approximation or finite-domain error.

\end{document}